\documentclass[12pt]{article}
\usepackage{hyperref}
\pdfoutput=1
\usepackage{amsmath,amssymb,amsthm} 
\usepackage{bbm}
\usepackage{stmaryrd}
\usepackage{paralist}
\usepackage{graphicx}
\usepackage{color}
\usepackage{comment}
\usepackage{todonotes}
\usepackage{tikz}
\usepackage{verbatim}
\usetikzlibrary{arrows,positioning} 
\tikzset{
    >=stealth',
    punkt/.style={
           rectangle,
           rounded corners,
           draw=black, very thick,
           text width=9.5em,
           minimum height=2em,
           text centered},
    imbox/.style={
           text width=11.5em,
           minimum height=2em,
           text centered},
    pil/.style={
           ->,
           thick,
           shorten <=2pt,
           shorten >=2pt,},
    pil2/.style={
           -,
           thick,
           shorten <=2pt,
           shorten >=2pt,},
    pil3/.style={
           <-,
           thick,
           shorten <=2pt,
           shorten >=2pt,}
}
\usepackage{mathtools}
\mathtoolsset{showonlyrefs}

\setdefaultenum{(i)}{}{}{}
\usepackage[margin=1in]{geometry} 

\numberwithin{equation}{section}

\renewenvironment{thebibliography}[1]{%
\begin{oldthebibliography}{#1}%
\setlength{\baselineskip}{.8em}
\linespread{.8}
\small 
\setlength{\parskip}{0ex}%
\setlength{\itemsep}{.1em}%
}%
{%
\end{oldthebibliography}%
}

\theoremstyle{plain}
\newtheorem{thm}{Theorem}[section]
\newtheorem{dfn}[thm]{Definition}
\newtheorem{prop}[thm]{Proposition}
\newtheorem{lem}[thm]{Lemma}

\theoremstyle{definition}
\newtheorem{rem}[thm]{Remark} 
\newtheorem{remark}[thm]{Remark} 
\newtheorem{example}[thm]{Example}

\def\be{\begin{eqnarray}}
\def\ee{\end{eqnarray}}
\def\b*{\begin{eqnarray*}}
\def\e*{\end{eqnarray*}}
\def\bee{\begin{equation}}
\def\eee{\end{equation}}

\DeclareMathOperator*{\argmax}{arg \, max}

\newcommand{\cA}{\mathcal{A}}

\newcommand{\cF}{\mathcal{F}}

\newcommand{\cH}{\mathcal{H}}

\newcommand{\cK}{\mathcal{K}}
\newcommand{\cL}{\mathcal{L}}

 \newcommand{\cS}{\mathcal{S}}

\newcommand{\ignore}[1]{}

\def \a{\alpha}
\def \b{\beta}
\def \e{\epsilon}
\def \d{\delta}

\def \E{\mathbb{E}}

\def \L{\mathbb{L}}

 \def\P{\mathbb{P}}

\DeclareMathAlphabet\scr{U}{scr}{m}{n}
\SetMathAlphabet\scr{bold}{U}{scr}{b}{n}
  \DeclareFontFamily{U}{scr}{\skewchar\font'177}%
  \DeclareFontShape{U}{scr}{m}{n}{<-6>rsfs5<6-8>rsfs7<8->rsfs10}{}%
  \DeclareFontShape{U}{scr}{b}{n}{<-6>rsfs5<6-8>rsfs7<8->rsfs10}{}%

\def\1{\mathbf{1}}
\begin{document}

\title{Liquidity in Competitive Dealer Markets\thanks{The authors thank Jan Kallsen and Roel Oomen for fruitful discussions. The extremely helpful suggestions of two anonymous referees are also gratefully acknowledged.}}

\author{Peter Bank\thanks{Technische Universit\"at Berlin, Institut f\"ur Mathematik, Stra{\ss}e des 17. Juni 136, 10623 Berlin, Germany, email \texttt{bank@math.tu-berlin.de}.}
\and
Ibrahim Ekren\thanks{Florida State University, Department of Mathematics, 1017 Academic Way, Tallahassee, FL 32306, email \url{iekren@fsu.edu}.}
\and
Johannes Muhle-Karbe\thanks{Imperial College London, Department of Mathematics, London SW71NE, UK, email \url{j.muhle-karbe@imperial.ac.uk}.  Research supported by the CFM-Imperial Institute of Quantitative Finance. Parts of this paper were written while this author was visiting the Forschungsinstitut f\"ur Mathematik at ETH Z\"urich.}
}

\date{March 2, 2021}

\maketitle

\begin{abstract}
We study a continuous-time version of the intermediation model of Grossman and Miller~\cite{grossman.miller.88}. To wit, we solve for the competitive equilibrium prices at which liquidity takers' demands are absorbed by dealers with quadratic inventory costs, who can in turn gradually transfer these positions to an exogenous open market with finite liquidity. This endogenously leads to transient price impact in the dealer market. Smooth, diffusive, and discrete trades all incur finite but nontrivial liquidity costs, and can arise naturally from the liquidity takers' optimization.
\end{abstract}




\textbf{}

\section{Introduction}

A basic paradigm of classical financial theory is that markets are
``perfectly liquid'', in that arbitrary amounts can be traded
immediately at the present market price. Yet, in reality liquidity is limited and generated by the interplay of strategic liquidity providers and consumers. This paper studies liquidity formation in a tractable equilibrium model, where intermediaries dynamically transfer liquidity from one market segment to another. 

More specifically, we consider a risky asset that is traded in two
markets. In the first market segment (the \emph{open market}), prices
and liquidity costs are given exogenously as in the optimal execution
literature following Almgren and Chriss~\cite{almgren.chriss.01}. In
the second market segment (the \emph{dealer market}), agents trade the
asset competitively. Some of the agents (who we will refer to as
\emph{clients}) have exogenously given trading needs, and are therefore willing to
pay a premium for the immediacy provided by the other agents (who we
will call \emph{dealers}).\footnote{Unlike the ``market makers'' in
  \cite{ho.stoll.81,kyle.85} and many more recent studies, these
  liquidity providers are not obliged to absorb any order flow, but
  trade at their discretion. Investment banks providing liquidity in
  foreign exchange markets are a typical example, compare
  \cite{butz.oomen.17}.}  The equilibrium price at which this second
market clears is in turn driven by the agents' aggregate demand as in
the model of Garleanu, Pedersen, and
Pothesman~\cite{garleanu.al.09}. The present study analyzes how agents
with access to both markets dynamically intermediate between
them. This extends the classical one-shot intermediation model of
Grossman and Miller~\cite{grossman.miller.88} to continuous time,
where inventory management and in turn equilibrium prices reflect a
tradeoff between past, present, and (expectations of) future order
flow. Our paper also contributes to the literature on segmented markets~\cite{gromb.vayanos.02,pavlova.rigobon.08,gromb.vayanos.10},
by analyzing a model where not only the positions of intermediaries
in the ``peripheral'' dealer market are constrained, but also position
adjustments in the ``central'' open market are costly.

Our stylized model is motivated by liquidity provision of competitive dealers in over-the-counter markets. As succinctly summarized by \cite{butz.oomen.17} in the context of foreign exchange markets,

\begin{itemize}
\item[] \emph{``Dealers in over-the-counter financial markets provide liquidity to customers on a principal basis and manage the risk position that arises out of this activity in one of two ways. They may internalize a customer's trade by warehousing the risk in anticipation of future offsetting flow, or they can externalize the trade by hedging it out in the open market.\footnote{This refers to inter-dealer broker platforms or public electronic crossing networks, for example.} [..] The notion that dealers are either perfect internalizers or perfect externalizers is of course too constraining and in practice they will use to varying extent a mix of both to manage their risk.''
}
\end{itemize}

In our model, pure internalizers only trade in the dealer market,
whereas externalization corresponds to offsetting trades in the open
market. Our dynamic model in turn allows to study the optimal tradeoff
between both risk management strategies. In order to permit the
analysis of this complex interaction between clients, dealers, and the
open market, we make a number of simplifying assumptions. First, all
agents act as price takers in the dealer market as
in~\cite{garleanu.al.09}. This means that the equilibrium price is
``competitive'', which is reasonable for a large number of small
liquidity providers, whose individual actions only have negligible
effects on the overall market equilibrium.\footnote{Evidently, a very
  important but challenging direction for further research is to study
  extensions to a game-theoretic setting, where dealers dynamically
  adjust their price quotes to account for their inventories while
  competing with each other for the clients' order flow. For
  liquidation problems, a first step in this direction is undertaken
  in \cite{capponi.al.18}, where dealers strategically quote bid-ask
  spreads for randomly arriving clients with exogenous demand curves,
  and a strategic client who needs to liquidate a large
  position. In general, however, non-competitive equilibrium prices can
  be formalized in many different ways and will support many different price
  dynamics
  (cf., e.g., \cite{sannikov.skrzypacz.16,choi.al.19}).} Second,
agents have quadratic inventory costs rather than preferences modelled
by concave utility functions. Such quadratic holding costs are also
used in
\cite{rosu.16,sannikov.skrzypacz.16,choi.al.19,muhlekarbe.webster.17,nutz.scheinkman.17},
for example, because this reduced-form modelling of ``inventory
aversion'' is considerably more tractable than risk aversion, yet
still penalizes the accumulation of large and thereby risky
positions. Third, the exogenous price in the open market has
martingale dynamics, complemented by quadratic trading costs on the
total and individual order flows. The martingale assumption, also made
in \cite{grossman.miller.88,garleanu.pedersen.16}, ensures that agents
do not speculate in the open market, but instead purely focus on
intermediation, in line with the empirical observation that ``specialists are good short-term traders but undistinguished long-term speculators''~\cite{hasbrouck.sofianos.93}.
Quadratic costs on the trading rate as
in~\cite{almgren.chriss.01,garleanu.pedersen.16} penalize trading
speed, modelling that positions can only be unwound gradually in the
open market. The individual trading costs allow to distinguish how
easy it is for different agents to access the open market; the
benchmark example is that dealers have low or no access costs whereas
retail clients have no direct access at all. The trading cost on the
total order flow reflects that it becomes more difficult to, for
instance, unwind a position in the open market when others are trying
to do the same. For finitely many agents, this introduces a negative
externality, where agents internalize the effects their trades have on
their own execution prices but not on others. In order to ensure
consistency with the competitive dealer market, we therefore identify
the limiting cases of our results for a large number of identical small
dealers. We find that the limiting model is similar to that of a representative
dealer with suitably aggregated parameters.

In the setting outlined above, our main result,
Theorem~\ref{thm:Nash}, identifies the unique equilibrium price in the
dealer market as the solution of a linear forward-backward stochastic
differential equation (FBSDE). Its explicit solution in terms of
hyperbolic functions and conditional expectations of the agents'
trading targets in turn reveals how liquidity is mitigated between the
dealer market and the open market, and how this affects prices and optimal trading strategies.  For the simplest case of an inelastic client order flow (i.e., only ``noise trades'') in the dealer market, we have
\begin{equation}\label{eq:price1}
\begin{split}
\mbox{Equilibrium Price } =& \mbox{ Fundamental Price }\\
 &\qquad+ \mbox{Holding Cost} \times \mbox{Expected Future Inventory.}
\end{split}
\end{equation}
This adjustment is consistent with the small-risk aversion limit
obtained for the model of \cite{garleanu.al.09}
in~\cite{kramkov.pulido.16a,kramkov.pulido.16}. However, whereas the
price impact is \emph{permanent} in these models without access to an
open market, it becomes \emph{transient} in our model, a feature of
price dynamics also documented in the empirical literature (cf., e.g., \cite{hasbrouck.sofianos.93} and the references therein). The reason is
that the client order flow can be gradually passed on to the open
market here, so that the expected future demand in the formula
of~\cite{kramkov.pulido.16} is replaced by the optimally controlled
inventory above. Put differently, the dealers in \cite{garleanu.al.09,kramkov.pulido.16,kramkov.pulido.16a} correspond to ``pure internalizers'' in the terminology of \cite{butz.oomen.17}, whereas our dealers employ both internalization and externalization in an optimal manner.

The liquidity costs implied by these equilibrium dynamics can be best understood in the case of a highly liquid open market. To wit, we show that as the trading costs in the open market tend to zero, the price in the dealer market converges to the same martingale price as in the open market. The clients' liquidity costs compared to this benchmark in turn admit simple, intuitive expressions that depend on the fluctuations of the client demand.  If the latter is smooth, trading through the dealer market is approximately equivalent to trading directly in the open market at a higher cost that depends on the number of dealers and reflects the premium that is necessary to entice the dealers to provide the necessary liquidity. This markup disappears in the limit of many small dealers, so that smooth client flow then trades at approximately the same prices as in the open market. This holds irrespective of the dealers' risk tolerances, since smooth flow can be hedged efficiently in the open market.

This changes for diffusive noise trades. Such irregular order flow is
more difficult to pass on to the illiquid open market. Accordingly,
the corresponding liquidity costs asymptotically scale with the square
root of the trading cost in the open market, multiplied by the
square root of the dealers' holding cost. The order flow at hand
enters through its quadratic variation, similarly as in the reduced
form model of \cite{cetin.al.04}. Here, however, the trading costs of
such ``rough'' strategies cannot be avoided as in
\cite{cetin.al.04,bank.baum.04} by approximating it with smooth
strategies. More generally, the liquidity costs implied by our model
are continuous in the client demand. As a consequence, strategies of
various forms are priced consistently and incur finite but nontrivial
liquidity costs which reflect the strategies' regularity. 

\medskip

The results described so far pertain to the case of a fixed given client demand as in~\cite{garleanu.al.09}. However, our model is tractable enough to also endogenize (some or all of) the order flow. More specifically, we can study how clients optimally track their target positions, so that their demand responds to prices. In this case, Formula~\eqref{eq:price1} remains valid, \emph{if} the future inventory is replaced by aggregate demand minus offsetting positions in the open market. If demand is inelastic, targets and actual positions coincide, and we recover the previous formula. In contrast, price-sensitive clients accept some deviations from their target positions, but the equilibrium price is still determined by the aggregate (expectations) of all individual target positions.

To illustrate the implications of this result and the corresponding optimal trading strategies, we consider two examples.  In the first, the clients' trading targets are constant as in~\cite{choi.al.19}. This leads to optimal liquidation problems similar to the ones studied in~\cite{bertsimas.lo.98,almgren.chriss.01,obizhaeva.wang.13} and many more recent papers.  Subsequently, we turn to an example with diffusive trading targets that correspond to the ``high-frequency trading needs'' considered in \cite{lo.al.04,sannikov.skrzypacz.16}.

For constant trading targets, the clients' optimal trading strategies are of a similar form as in optimal liquidation models with transient price impact \cite{obizhaeva.wang.13}: isolated bulk trades combined with otherwise smooth order flow. In contrast, diffusive trading targets as in \cite{lo.al.04,sannikov.skrzypacz.16} lead to optimal client demands with nontrivial quadratic variation. Therefore, our model consistently combines the qualitative properties of standard models for optimal liquidation, while nevertheless allowing for rapidly fluctuating inventories in line with the empirical evidence documented by \cite{carmona.webster.13}, for example.

The effects of limited liquidity on the equilibrium price also depends on the properties of the client order flow. In the optimal liquidation example (and, more generally, for deterministic targets), illiquidity only affects price levels and expected returns, but not volatility. In contrast, random demands such as the diffusive trading targets generally also change volatility. Here, the sign of this change is determined by the correlation between trading targets and the fundamental value similarly as in~\cite{herdegen.al.19}. 

\medskip

This article is organized as follows. Our model for
the dealer and open markets is introduced in
Section~\ref{s:model}. Section~\ref{s:equilibrium} in turn contains
our characterization of equilibrium prices and trading strategies in
this context, as well as a detailed discussion of the implications of
these results. For better readability, all proofs are delegated to the appendix.

\paragraph{Notation} 
A predictable process $X=(X_t)_{t \in [0,T]}$ belongs to $\mathcal{L}^2$ if $\mathbb{E}[\int_0^T X_t^2dt]<\infty$ and to $\mathcal{S}^2$ if $\mathbb{E}[\sup_{t \in [0,T]} X^2_t]<\infty$. The set of square-integrable martingales is denoted by $\mathcal{M}^2$, and we write $\mathcal{H}^2$ for the semimartingales whose local martingale part belongs to $\mathcal{M}^2$ and whose finite-variation part has square-integrable total variation. For an It\^o process with dynamics $dX_t=\mu_t dt+\sigma dW_t$, this holds if $\mathbb{E}[(\int_0^T |\mu_t|dt)^2+\int_0^T\sigma_t^2dt]<\infty$.

\section{Model}\label{s:model}

\subsection{Agents}
We consider finitely many agents indexed by
$a \in \cA \not= \emptyset$. Agent $a$ has mass $m(a) \in (0,\infty)$, and all agents share the same beliefs and
information flow described by the filtered probability space
$(\Omega,\cF,(\cF_t),\P)$ satisfying the usual conditions of right-continuity and completeness.

\subsection{Financial Market}
The agents invest in two assets. The first is riskless
 and bears no interest (for simplicity). The second asset pays an exogenous liquidating dividend
$\scr{D}_T \in L^2(\cF_T)$ at time $T>0$ and can be traded in two markets. 

\paragraph{Dealer Market}
In the first market, the agents $a \in \mathcal{A}$ competitively trade the asset at the \emph{equilibrium price}  
$S \in \cH^2$ which clears this market and matches the terminal payoff $S_T=\scr{D}_T$. In our main result, Theorem~\ref{thm:Nash}, we determine this price explicitly.

If some of the agents have exogenous trading needs whereas others don't, this market can be interpreted as a competitive dealer market, where the 
``dealers'' (or ``intermediaries'' or ``market makers'') from the second group earn a premium for supplying liquidity to the first group of ``clients'' (or ``outside customers''~\cite{grossman.miller.88} or ``end-users''~\cite{garleanu.al.09}). 

In general, if agent $a\in\cA$ follows a trading strategy $K^a \in \cS^2$ (specifying the number of risky assets held) in the dealer market, then this generates the expected P\&L
\begin{align}
  \label{eq:2}
  \E\left[\int_0^T K^a_t dS_t\right] = \E\left[\int_0^T K^a_t dA_t\right].
\end{align}
Here, $A$ denotes the ``drift'' of $S$, i.e., the predictable bounded variation part in its Doob-Meyer-decomposition.

\paragraph{Open Market}

The dealer market described above models liquidity provision by ``pure internalizers'', who absorb their customers' order flow until it is offset by future incoming orders.  In real dealer markets, this is complemented by ``externalization'', i.e., actively ``hedging out trades in the open market''~\cite{butz.oomen.17}. In the classic model of Grossman and Miller~\cite{grossman.miller.88} (or in the recent work of Garleanu and Pedersen~\cite[Section~3.2]{garleanu.pedersen.16}, for example) the risk inherent in this intermediation process is linked to a fixed search time needed to locate a counterparty that allows the dealers to lay off their positions. In the present study, we formulate and solve a \emph{dynamic} version of the dealers' inventory management problem, where the risky asset can be traded at all times on the open market, but with a trading friction that imposes a penalty for the fast liquidation of large positions. 

To wit, in addition to the dealer market, we consider a second market for the risky asset, where its unaffected price process is given by its expected dividend,
\begin{align}
  \label{eq:3}
  \scr{D}_t := \E_t[\scr{D}_T], \quad 0 \leq t \leq T.
  \end{align}
This martingale price is assumed for simplicity, because it eliminates speculation in the open market. The corresponding execution price is in turn given by $\scr{D}_t+\lambda \bar{u}+\frac12 \lambda^a u^a$ for the trading rate $u^a \in \cL^2$ of agent $a \in \cA$, that is, the speed at which her current position $U^a_t=\int_0^t u^a_s ds$ in the open market is adjusted. (The position $U^a$ then automatically belongs to $\cS^2$.) Here, $\bar{u} := \sum_{a \in \cA} m(a) u^a$ and the constant 
$\lambda \geq 0$ describes the impact that this total order flow has on the market price. In contrast, the parameter $\lambda^a$ models agent $a$'s idiosyncratic trading frictions such as search costs, that do not depend on the other agents' trading activities.\footnote{This is similar in spirit to allowing agents to pay a quadratic cost to shorten their search times.} If the open market models interdealer trading, then the individual frictions are naturally finite for the dealers but infinite for their customers.

 In summary, agent $a$'s trades in the open market generate an expected P\&L of
\begin{align}
  \nonumber
  \E&\left[\scr{D}_TU^a_T -\int_0^T \left(\scr{D}_t+\lambda \bar{u}_t\right) u^a_t dt -\int_0^T \frac{\lambda^a}{2}
  (u^a_t)^2 dt\right]\\\label{eq:4}
 &= -\E\left[\int_0^T \left(\lambda u^{-a}_t u^a_t
   +\left(\lambda m(a)+\frac{1}{2}\lambda^a\right) (u^a_t)^2\right)dt\right].
\end{align}
Here, $\bar{u}^{-a} := \bar{u}-m(a)u^a$ denotes the other agents' aggregate
trading rate.

\subsection{Goal Functionals}

We now formulate the agents' goal functionals. Similarly as in \cite{garleanu.pedersen.13,garleanu.pedersen.16,sannikov.skrzypacz.16,bouchard.al.17,choi.al.19}, these are chosen to be of a linear-quadratic form for maximum tractability.

\paragraph{Trading Targets and Inventory Costs}

The agents have an incentive to trade due to exogenous exposures to risk.\footnote{The impact of illiquidity on trades due to heterogenous beliefs about fundamentals is studied in~\cite{muhlekarbe.al.20}.} A particularly convenient way to model this is to
fix a \emph{target position} $\xi^a \in \cS^2$ for each
agent $a \in \cA$, and in turn penalize the mean-squared distance to the agent's total risky position (accumulated either from dealers or in the open market):
\begin{align}
  \label{eq:5}
 \frac{1}{2} \E\left[\int_0^T (\xi^a_t-K^a_t-U^a_t)^2dt \right].
\end{align}
The importance of the inventory penalty~\eqref{eq:5} relative to expected trading gains and losses \eqref{eq:2}, \eqref{eq:4} is described by the agents' \emph{risk-tolerances} $\rho^a>0$, $a\in \cA$. Together, this leads to the linear-quadratic goal functional
\begin{align}
  \label{eq:7}
  J^a(K^a,u^a;\bar{u}^{-a},S) := \E&\left[\int_0^T K^a_t dA_t - \int_0^T \left(\lambda u^{-a}_t u^a_t
   +\left(\lambda m(a)+\frac{1}{2}\lambda^a\right) (u^a_t)^2\right)dt\right] \\
   &-\E\left[\int_0^T \frac{1}{2\rho^a}(\xi^a_t-K^a_t-U^a_t)^2dt \right].
\end{align}
Given the other agents' aggregate transactions $\bar{u}^{-a}$ in the open market and a price process $S \in \cH^2$ in the dealer market, agent $a$ maximizes this over $K^a \in \cS^2$ and $u^a \in \cL^2$. Our goal now is to determine the \emph{equilibrium} price $S$ for which the dealer market clears.

\section{Equilibrium}\label{s:equilibrium}

The agents interact with each other through the equilibrium price in the dealer market and their common price impact in the open market. With additional exogenous demands 
$K^N \in \cS^2$ from noise traders in the dealer market, this leads to the following notion of equilibrium:

\begin{dfn}
  A price process $S \in \cH^2$ in the dealer market and trading strategies
  $(K^a)_{a \in \cA}$ and $(u^a)_{a \in \cA}$ in the dealer and the open markets, respectively, form an
  \emph{equilibrium} if
\begin{enumerate}
\item[(i)] the dealer market clears, in that $K^N+\sum_{a \in \cA} m(a) K^a=0$;
\item[(ii)] the trading strategies form an (open loop) Nash
  equilibrium in that the strategy $(K^a,u^a)$ maximizes agent $a$'s target
  functional $J^a(\cdot,\cdot;\bar{u}^{-a},S)$ over $\cS^2 \times \cL^2$.
\end{enumerate}
\end{dfn}

\begin{remark}
As in \cite{garleanu.al.09}, we study a competitive equilibrium in the dealer market for tractability, and also because non-competitive equilibrium prices in continuous time necessarily introduce additional degrees of freedom (cf., e.g., \cite{sannikov.skrzypacz.16,choi.al.19}). 

The price-taking assumption in competitive equilibria is reasonable for a large number of small agents. In this regime, however, it is challenging to model the trading frictions in the open market in a consistent and nontrivial manner. For example, if each (small) agent simply neglects their own contribution to the execution price, this leads to a linear trading cost and in turn large trading rates that are not negligible even for small agents. To resolve this, we start from Nash competition between finitely many agents and in turn identify the competitive limit for many small agents in a second step, cf.~the discussion after Theorem~\ref{thm:Nash}.
\end{remark}

We now present our main result on existence and uniqueness result of an equilibrium price for the dealer market. 
To formulate this concisely, it is convenient to recall the following result from~\cite{bouchard.al.17}[Theorem A.4]:

\begin{lem}\label{lem:fbsde}
Fix $X \in \cL^2$ and $\Delta>0$. Then, the unique solution $(u,U) \in \cL^2 \times \cH^2$ of the following
linear forward-backward stochastic differential equation (FBSDE)\footnote{Here, the square-integrable martingale $M$ is part of the solution. }
\begin{equation}
\begin{split}
  \label{eq:25}
  du_t &= \Delta \left(U_t-X_t \right)dt+dM_t, \quad u_T=0,\\
  U_0=0, \quad dU_t &= u_t dt,
  \end{split}
\end{equation}
is given by 
\begin{align}
  \label{eq:30}
  \mathbf{u}^\Delta_t(X)&:={\E_t\left[\int_t^T k^\Delta(t,s)X_s ds\right] }-F^\Delta(t)\mathbf{U}^\Delta_t(X), \quad
       0 \leq t \leq T,\\ 
  \mathbf{U}^\Delta_t(X)&:=\frac{1}{\Delta}\int_0^t k^\Delta(s,t) \E_s\left[\int_s^T k^\Delta(s,r)X_r dr\right] ds,\label{eq:33} \quad
       0 \leq t \leq T,
       \end{align}
for the kernel and the function
\begin{align}
  \label{eq:10}
   k^\Delta(t,s)=
  \frac{\Delta\cosh(\sqrt{\Delta}(T-s))}{\cosh(\sqrt{\Delta}(T-t))},\quad F^\Delta(t)=\sqrt{\Delta}\tanh\left(\sqrt{\Delta} (T-t)\right),
  \quad 0 \leq s,t \leq T.
\end{align}
\end{lem}

To rule out frictionless trading in the open market, we assume that
\begin{align}
  \label{eq:6}
   \lambda + \lambda^a>0, \quad a \in \cA,
\end{align}
which allows us to introduce the \emph{price elasticities}
$$\eta^a =1/(m(a)\lambda+\lambda^a), \quad  \bar{\eta} = \sum_{a \in \cA} m(a)\eta^a, \quad  \eta=1/\lambda.$$ 
(In the case $\lambda=0$, it is natural to let $\eta=+\infty$ and $\Delta= \bar{\eta}/\bar{\rho}$.)
We also define the agents' aggregate risk tolerance and target position:
$$\bar{\rho}=\sum_{a\in\cA} m(a)\rho^a, \quad \bar{\xi} = \sum_{a \in \cA} m(a)\xi^a.$$
With these ingredients, we can now formulate our main result, which characterizes the unique equilibrium price in the dealer market as well as the agents' optimal trading strategies in both the dealer market and the open market.

\begin{thm}\label{thm:Nash}
\begin{enumerate}
\item[(i)]
  There exists a unique equilibrium
  $(S,(K^a,u^a)_{a \in \cA}) \in \cH^2 \times (\cS^2 \times \cL^2)^\cA$. 
  \item[(ii)]  The agents' aggregate position $\bar{U}=\sum_{a \in \mathcal{A}} m(a)U^a$ in the open
  market is 
  \begin{align}
  \label{eq:8}
    \bar{U}_t =  \mathbf{U}_t^\Delta\left(K^N+\bar{\xi}\right), \quad 0 \leq t \leq T,
  \end{align}
   for $\mathbf{U}^\Delta$ from Lemma~\ref{lem:fbsde} with 
   $$\Delta= \frac{1}{\bar{\rho}} \frac{\eta \bar{\eta}}{\eta+\bar{\eta}}.$$
  Agent $a$'s share of $\bar{U}_t$ is
  $U^a_t = \frac{\eta^a}{\bar{\eta}} \bar{U}_t$. 
  \item[(iii)] The equilibrium price of the risky asset in the dealer market is
  \begin{align}
    \label{eq:13}
    S_t = \E_t\left[\scr{D}_T-\int_t^T \mu_s ds\right], \quad \mbox{where }  \mu_t = 
    \frac{1}{\bar{\rho}}\left(\bar{U}_t-K^N_t-\bar{\xi}_t\right),
    \quad 0 \leq t \leq T. 
      \end{align}
  Agent $a$'s optimal position in the dealer market is
  \begin{align}
    \label{eq:15}
      K^a_t = \xi^a_t -\frac{\eta^a}{\bar{\eta}}\bar{U}_t +\frac{\rho^a}{\bar{\rho}}\left(\bar{U}_t-K^N_t-\bar{\xi}_t\right),
    \quad 0 \leq t \leq T.
  \end{align}
  \end{enumerate}
  \end{thm}
  
  \begin{proof}
See Appendix~\ref{s:proof}.
\end{proof}

The risk premium $\mu_t$ in the dealer market is thus
determined by the agents' aggregate exposure $K^N+\bar{\xi}-\bar{U}$
at any one time, measured in units of their aggregate risk tolerance
$\bar{\rho}$. In view of~\eqref{eq:15}, the risk premium incentivizes agent $a$
to accept her share of this exposure in proportion to her individual risk
tolerance $\rho^a$, having contributed to the aggregate risk-transfer
to the open market $\bar{U}$ in proportion with her effective
elasticity $\eta^a$.  The asset price \emph{fluctuations} in the dealer market (i.e., the martingale part of the asset price~\eqref{eq:13}) in turn depend on how uncertainty about the
liquidating dividend $\scr{D}_T$ and future risk premia are revealed over time.

Alternatively, the first-order conditions for the optimal trading
rates in the open market show that the equilibrium price from Theorem~\ref{thm:Nash} admits the following concise representation:
\begin{equation}\label{eq:price}
S_t=\scr{D}_t+\left(\frac{1}{\eta}+\frac{1}{\bar \eta}\right){\bar u_t}.
\end{equation}
This means that the  adjustment of the equilibrium price  compared to
expected dividend is determined by the agents' aggregate trading rate
in the open market measured in units of the aggregate price elasticity
$1/(\eta^{-1}+\bar{\eta}^{-1})$ resulting from the combination of the
agents' idiosyncratic and their common impact on open market
prices. To wit, suppose the agents are on aggregate buying in the
open market ($\bar{u}_t>0$) because they want to increase their net
position. Then they will also be willing to purchase further risky
assets in the dealer market at a premium. Despite this appealingly simple interpretation, the  dependence of the equilibrium price on  the  model  parameters  is  generally  rather  involved, since it typically depends on the past, present, and (expectations of) future demands.  

\subsection{Large-Liquidity Asymptotics for Noise Trades}\label{ss:asym}

In order to better understand the equilibrium prices~\eqref{eq:13} and their implications, we therefore first consider the simplest version of the dealer market, where $M \in \mathbb{N}$ ``dealers'' (i.e., agents without idiosyncratic trading targets, $\xi^d=0$) with common  masses $m(a)=1/M$, risk tolerances $\rho^a=\rho_d>0$, and individual trading costs $\lambda^a=0$ absorb the demand $K^N\in \mathcal{S}^2$ of noise traders. 

Our first result shows that the equilibrium price~\eqref{eq:13} approaches the expected dividend as the open market becomes more and more liquid for $\lambda \to 0$: 

\begin{prop}\label{prop:priceasym}
For any noise-trader demand $K^N \in \mathcal{S}^2$, the equilibrium price $S$ from
Theorem~\ref{thm:Nash} converges to the expected dividend $\scr{D}_t=\mathbb{E}_t[\scr{D}_T]$ in the particularly strong sense that
$$
\sup_{-1 \leq H \leq 1 \text{ predictable }} \E \left[\sup_{0 \leq t \leq T}
\left|\int_0^t H_s d(\scr{D}_s-S_s)\right|^2\right] \to 0, \quad \text{ as $\lambda \to 0$.}
$$
In particular $S$ converges to $\scr{D}$ in the Emery topology as $\lambda \to 0$. Moreover, the corresponding wealth processes generated by the noise traders' demand satisfy
$$
\int_0^T K^N_t dS_t = \int_0^T K^N_t d\scr{D}_t +o(1), \quad \mbox{in $L^1$ as $\lambda \to 0$.}
$$
\end{prop}
Proposition~\ref{prop:priceasym} asserts that, as the open market becomes more and more liquid, the dealer price approaches the expected payoff of the risky asset. This is intuitive because with vanishing trading frictions in the open market, dealers can immediately unload their inventories  at this price. 

Given additional structure of the noise-trader demand, we can also identify the leading-order correction term for the noise traders' wealth process, i.e., the liquidity costs implied by the dealers' nontrivial but finite risk-bearing capacity. The form of this leading-order correction term depends on the variability of the clients' demand. To illustrate this, we discuss two examples that appear frequently in applications: smooth demands $K^N_t=\int_0^t \mu^N_s ds$ and diffusive demands with It\^o dynamics $K^N_t=\int_0^t \mu^N_s ds+\int_0^t \sigma^N_s dW_s$.

\paragraph{Smooth Demands} We first discuss demands that accumulate at a finite, absolutely continuous rate. In this case, the dealers could hedge their exposure perfectly by passing on their positions immediately to the open market subject to the quadratic cost $\lambda$ imposed on the corresponding trading rate. Therefore, the dealers could break even using this strategy for an execution price equal to the fundamental value plus the trading cost in the open market. However, due to the quadratic nature of the trading cost, they could achieve strictly positive profits in this case by absorbing only a fraction of the clients' demands. Accordingly, in equilibria with finitely many dealers, the dealers need to be paid an additional premium, but the latter vanishes in the limit of many small dealers.

More specifically, the subsequent result explicitly identifies the leading-order term of this liquidity cost as $(M+1)/M$ times the trading cost in the open market. This term is independent of the dealers' inventory costs, since smooth client demands can be hedged very efficiently by trading in the open market.

\begin{lem}\label{lem:costsasym1}
Suppose that $K^N_t=\int_0^t\mu^N_sds$ for a continuous process $\mu^N \in \scr{S}^2$. Then, the liquidity costs generated by $K$ are
\begin{equation}\label{eq:limit1}
\int_0^T K^N_t d\scr{D}_t -\int_0^T K^N_tdS_t  =\lambda\frac{M+1}{M} \int_0^T\left(\mu^N_t\right)^2 dt+o(\lambda), \quad \mbox{in $L^1$~as $\lambda \to 0$.}
\end{equation}
\end{lem}

In particular, Lemma~\ref{lem:costsasym1} implies that the equilibrium trading costs in \emph{competitive} dealer markets with many small dealers are approximately the same as in the open market.

\paragraph{Diffusive demands}

Next, we turn to trading strategies with nontrivial Brownian fluctuations. These could not be implemented directly in the open market, but can be traded at a finite cost through the dealers. Since the dealers can hedge these more irregular order flows less efficiently than the smooth flows considered above, the corresponding trading costs are of a higher asymptotic order, namely $O(\sqrt{\lambda})$ instead of $O(\lambda)$ as $\lambda \to 0$. Moreover, the dealers' inventory cost now becomes visible in the leading-order term. The asymptotically crucial feature of client demand turns out to be its quadratic variation, which also appears in the reduced-form model of \cite{cetin.al.04}, for example:\footnote{The same scaling and the target strategy's quadratic variation also appear if such diffusive target positions are tracked optimally in markets with quadratic costs in the trading rate, cf.~\cite{moreau.al.17} and the references therein.}

\begin{lem}\label{prop.costsemi}
Suppose that the underlying filtration is generated by a Brownian motion $W$ and assume that the noise-trader demand has It\^o dynamics,
$$K^N_t=\int_0^t\mu^N_sds+\int_0^t\sigma^N_s dW_s, \quad 0 \leq t \leq T.$$
Here, $|K^N|^2, |\mu^N|^2,|\sigma^N|^2 \in \mathcal{H}^2$ and these processes are Malliavin differentiable in the sense of~\cite[p.~27]{nualart2006malliavin}): $K^N_t$, $\mu^N_t,\sigma^N_t\in \mathbb{D}^{1,2}$, with continuous Malliavin derivatives  $$s \mapsto \left(D_t\left(K^N_s\right),D_t\left(\mu^N_s\right),D_t\left(\sigma^N_s\right)\right),\,0 \leq t \leq s \leq T.$$
 Finally, suppose that $\sup_{0\leq t\leq T}\E[\sup_{ t\leq s\leq T}(|(D_t(K^N_s))|^2+|(D_t(\mu^N_s))|^2)]<\infty$.
Then, the liquidity costs generated by the demand $K^N$ are
\begin{equation}\label{eq:limit2}
\int_0^T K^N_t d\scr{D}_t -\int_0^T K^N_tdS_t = \sqrt{\frac{\lambda}{\rho_d}\frac{M+1}{M}  } \int_0^T(\sigma^N_t)^2 dt+o(\sqrt{\lambda}), \quad \mbox{in $L^1$ as $\lambda \to 0$.}
\end{equation}
\end{lem}

\begin{rem}
The regularity conditions of Proposition~\ref{prop.costsemi} are satisfied, in particular, if the demand $K^N$ is the solution of a scalar stochastic differential equation whose drift and diffusion coefficients are twice continuously differentiable with bounded derivatives of orders $0,1,2$. In this case, the required bounds for the Malliavin derivatives follow from \cite[Theorem 2.2.1]{nualart2006malliavin}. 
\end{rem}

\subsection{Competitive Dealer Markets}\label{ss:segmented}

Next, we apply Theorem~\ref{thm:Nash} to study liquidity in a competitive dealer market that we model by a large number of small homogenous dealers and clients. For simplicity, suppose that there are no noise traders ($K^N=0$), but $m_d M \in \mathbb{N}$ dealers with common risk tolerance $\rho_d$ and $m_c M$ clients with common risk tolerance $\rho_c>0$, all with equal mass $1/((m_d+m_c)M)$. That is, the dealers and clients make up fractions
$$q_d=\frac{m_d}{m_d+m_c} \quad \mbox{and} \quad q_c=\frac{m_c}{m_c+m_d}$$ 
of the total number $(m_d+m_c)M$ of agents. This allows us to study the limiting behaviour of equilibrium prices and trading strategies as the number of agents becomes large for $M \to \infty$, while the fractions of dealers and clients remains fixed.\footnote{We are grateful to an anonymous referee for prompting us to pursue this extension with general rather than equal proportions of dealers and clients.} The clients have a common trading target $\xi^c \in  \cS^2$, whereas the dealers have no trading targets ($\xi^d=0$) and therefore only trade to earn premia for providing liquidity. 

As is natural for an open market that describes interdealer trading, we assume that the individual trading frictions are zero for all dealers ($\lambda^a=0$) and infinite for the customers ($\lambda^a=\infty$). In summary, we then have 
\begin{gather*}
\bar{\rho}=q_c \rho_c+q_d\rho_d, \quad \bar{\eta}=\frac{m_d M}{\lambda}, \quad   \bar{\xi}=q_c\xi^c,\quad \Delta=\frac{1}{(q_c\rho_c+q_d\rho_d)\lambda(1+\frac{1}{m_d M})}.
\end{gather*}
As the number of dealers and clients becomes large for $M \to \infty$ (with the proportions $q_d$ and $q_c$ of dealers and clients remaining fixed), $\Delta$ therefore converges to a nonzero and finite limit,
$$\Delta_\infty=\frac{1}{(q_c\rho_c+q_d\rho_d)\lambda}.$$

Let us now compare this to a market with the same proportions $q_c$ of clients and $q_d$ of dealers, but where the $m_d M$ small dealers are replaced by a single dealer with the same aggregate mass $q_d$ and risk tolerance $\rho_d$. Then the above expressions for $\bar{\rho}$, $\bar{\xi}$ remain unchanged, but we have $\bar{\eta}=1/\lambda$ and in turn
$$
\Delta_1=\frac{1}{(q_c\rho_c+q_d\rho_d)2\lambda}.
$$
As the equilibrium prices in Theorem~\ref{thm:Nash} only depend on the aggregate trading target $\bar{\xi}$ (which is the same in both cases) and the parameter $\Delta$, we see that the market populated by many small dealers is equivalent to a market with a single representative dealer, if the liquidity cost in the competitive market is rescaled by a factor of two. Accordingly, the positive effects (for the clients) of competition (which drives down each dealer's profits) outweigh the negative effects of uncoordinated trading in the open market (which leads to excess trading because agents don't internalize the price impact costs they cause for others).

To illustrate the implications of these results on optimal trading strategies and equilibrium prices, we now specialize the discussion to two concrete examples: optimal execution as in \cite{bertsimas.lo.98,almgren.chriss.01,obizhaeva.wang.13} and diffusive trading targets as in \cite{lo.al.04,sannikov.skrzypacz.16,choi.al.19}.

\begin{example}[Optimal Liquidation]\label{ex:liq}
We first consider the simplest example where the clients' target is to
sell a certain number of shares, that is, $\xi^c_t \equiv \xi^c <0$, $0 \leq t \leq T$.\footnote{Complete liquidation as in \cite{bertsimas.lo.98,almgren.chriss.01,obizhaeva.wang.13} could be promoted using a quadratic liquidation penalty as in \cite{cartea.jaimungal.16,bank.voss.16} or enforced by a hard terminal constraint as in \cite{bank.al.17}. To ease notation, we do not pursue this here.} In this case, two elementary integrations show that 
\begin{align*}
\bar{U}_t=\mathbf{U}^{\Delta}_t(\bar{\xi}) &=\frac{1}{\Delta}\int_0^t k^{\Delta}(s,t)\int_s^T \left(k^{\Delta}(s,r) q_c\xi^c\right) dr  ds \\
&= \sqrt{\Delta}\cosh\left(\sqrt{\Delta}(T-t)\right)\int_0^t \frac{\sinh(\sqrt{\Delta}(T-s))}{\cosh^2(\sqrt{\Delta}(T-s))}ds\ q_c\xi^c\\
&= \left(1-\frac{\cosh(\sqrt{\Delta}(T-t))}{\cosh(\sqrt{\Delta}T)}\right) q_c \xi^c.
\end{align*}
As a consequence, the optimal client position from~\eqref{eq:15} is
\begin{align*}
K^c_t &= \xi^c+\frac{\rho_c}{q_c\rho_c+q_d\rho_d}\left(\bar{U}_t-q_c\xi^c\right)\\
&=\xi^c\left(1 -\frac{q_c\rho_c}{q_c\rho_c+q_d\rho_d}\frac{\cosh(\sqrt{\Delta}(T-t))}{\cosh(\sqrt{\Delta}T)}\right).
\end{align*}
This means that the clients use a bulk trade at time $t=0$ to sell a fraction of their trading target equal to their share $q_c\rho_c$ of the total holding costs $q_c\rho_c+q_d\rho_d$. With an open interdealer market, they subsequently continue selling at an absolutely continuous rate, as the dealers gradually pass on their positions. Whence, the clients' optimal trading path resembles the one in the model of Obizhaeva and Wang~\cite{obizhaeva.wang.13} with transient price impact. The absolutely continuous trading rate is determined by the trading costs $\lambda$ in the open market, dealers' and clients' share of the total risk tolerance, and the numbers of dealers and clients through the constant $\Delta$, with low trading costs, low risk tolerances, and a large number of dealers lead to faster trading. This is illustrated in the left panel of Figure~\ref{fig1}.

\begin{figure}
  \centering
  \includegraphics[width=.45\linewidth]{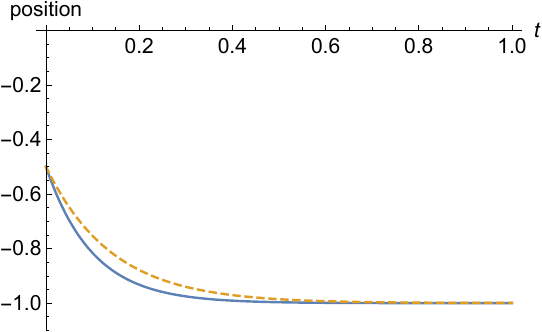}
  \includegraphics[width=.45\linewidth]{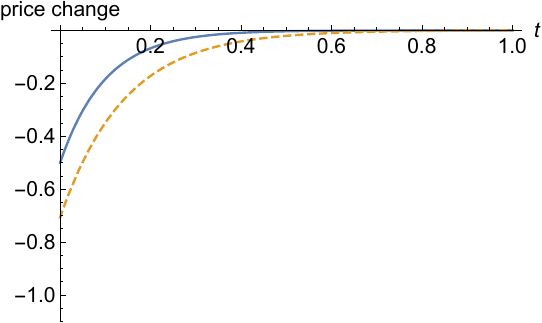}
\caption{Optimal liquidation strategies (left panel) and price changes (right panel) for a single dealer (dotted) and infinitely many dealers (solid). Parameters are $\lambda=0.1$, $\rho_c=\rho_d=0.1$, $T=1$, $\xi_c=-1$, and the dealers have half the total risk tolerance in each case.}
\label{fig1}
\end{figure}

Let us now turn to the equilibrium price for the dealer market. By Theorem~\ref{thm:Nash}(iii), 
\begin{align*}
S_t = \E_t[\scr{D}_T]-\frac{1}{\bar{\rho}}\E_t\left[\int_t^T (\bar{U}_s-\bar{\xi}_s) ds\right] &=\scr{D}_t +\frac{q_c\xi^c}{q_c\rho_c+q_d\rho_d} \int_t^T \frac{\cosh(\sqrt{\Delta}(T-s))}{\cosh(\sqrt{\Delta}T)}ds\\
&= \scr{D}_t +\frac{q_c\xi^c}{q_c\rho_c+q_d\rho_d}  \frac{\sinh(\sqrt{\Delta}(T-t))}{\sqrt{\Delta}\cosh(\sqrt{\Delta}T)}.
\end{align*}
With customers that want to liquidate a position in the dealer market
($\xi^c<0$), the risky asset trades at a price below its expected
dividend. The corresponding positive risk premium is earned by the
dealers for providing liquidity to the clients as the latter liquidate their position. Like for the clients' optimal positions the transience of the price deviation from the asset's expected payoff is modulated by the constant $\Delta$. Without an open market ($\lambda=\infty$ so that $\Delta=0$), the price impact is permanent; as the open market becomes more and more liquid, the price impact of the initial bulk trade disappears faster and faster as dealers quickly unwind their positions. Moreover, as illustrated in the right panel of Figure~\ref{fig1}, the price impact is decreasing in the number of dealers when their proportion of the total risk tolerance in the economy is held fixed. 

In this example, the volatility of the asset price in the dealer market remains unaffected because the clients' demand is deterministic. This will be different in the subsequent example with random demands. 
\end{example}

\begin{example}[Diffusive Trading Targets]\label{ex:dif}
To explore the other end of the spectrum of potential target strategies, suppose as in \cite{lo.al.04,sannikov.skrzypacz.16,choi.al.19} that the clients have ``high-frequency trading needs''. This is modelled by a target position $\xi^c$ following Brownian motion with volatility $\sigma_\xi$. For such a martingale, we have
\begin{align}
\bar{U}_t=\mathbf{U}^{\Delta}_t(\bar{\xi}) &=\frac{1}{\Delta}\int_0^t k^{\Delta}(s,t)\int_s^T \left(k^{\Delta}(s,r) q_c\xi_s^c\right) dr  ds  \notag\\
&= q_c \sqrt{\Delta}\cosh\left(\sqrt{\Delta}(T-t)\right)\int_0^t \left(\frac{\sinh(\sqrt{\Delta}(T-s))}{\cosh^2(\sqrt{\Delta}(T-s))} \xi_s^c\right)ds. \label{eq:exU}
\end{align}
As a consequence, the clients' optimal position from~\eqref{eq:15} is a convex combination of a diffusive and a smooth component,
$$
K^c_t= \frac{q_d\rho_d}{q_c\rho_c+q_d\rho_d}\xi^c_t+\frac{q_c\rho_c}{q_c\rho_c+q_d\rho_d} \sqrt{\Delta}\cosh\left(\sqrt{\Delta}(T-t)\right)\int_0^t \frac{\sinh(\sqrt{\Delta}(T-s))}{\cosh^2(\sqrt{\Delta}(T-s))}\xi_s^c ds.
$$
To wit, the clients directly implement a fraction of their trading target equal to the dealers' share of the total risk tolerance. To reduce their remaining risk, they also trade in an absolutely continuous-manner. This gradually pulls the remaining deviation from their trading targets towards zero, as can be seen from the dynamics
$$
dK^c_t= F^{\Delta}(t) \left(\xi^c_t -K^c_t\right)dt+\frac{q_d\rho_d}{q_c\rho_c+q_d\rho_d}d\xi^c_t.
$$
These trading strategies are illustrated in the left panel of Figure~\ref{fig2}, which shows that the number of dealer only has a modest impact on the clients' optimal positions relative to their target positions.

\begin{figure}
  \centering
  \includegraphics[width=.45\linewidth]{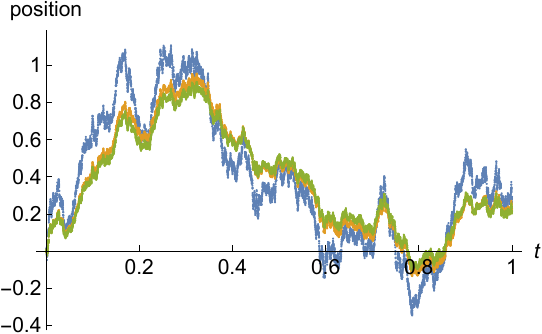}
    \includegraphics[width=.45\linewidth]{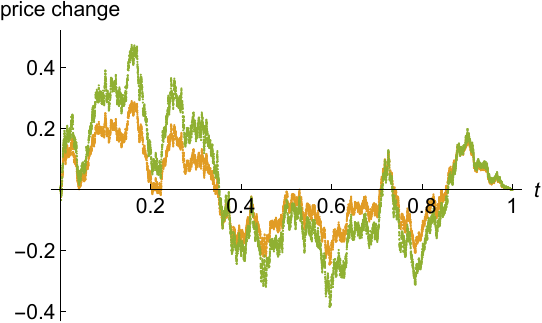}
\caption{Left panel: simulated target position (blue) and corresponding optimal position in dealer market with infinitely many small dealers (orange) and a single dealer (green), with half the total risk tolerance in each case. Right panel: Price adjustments in equilibrium with infinitely many small dealers (orange) and a single dealer (green). Parameters are $\lambda=0.1$, $\rho_c=\rho_d=0.1$, $T=1$, and $\sigma_\xi=1$.}
\label{fig2}
\end{figure}

Let us now turn to the corresponding equilibrium prices in the dealer market. In view of Theorem~\ref{thm:Nash}(iii), we have
\begin{align}\label{eq:priceex}
S_t &= \scr{D}_t +\frac{1}{\bar{\rho}}\E_t\left[\int_t^T (\bar{\xi}_s-\bar{U}_s) ds\right].
\end{align}
After recalling~\eqref{eq:exU}, differentiation shows that the process $\bar{\xi}_t-\bar{U}_t$ has Ornstein-Uhlenbeck-type dynamics,
$$
d(\bar{\xi}_t-\bar{U}_t)= -F^{\Delta}(t)(\bar{\xi}_t-\bar{U}_t)dt+q_c d\xi^c_t.
$$
As a consequence, the price adjustment $S_t-\scr{D}_t$ relative to the expected terminal dividend is 
\begin{align*}
\frac{1}{\bar{\rho}}\E_t\left[\int_t^T (\bar{\xi}_s-\bar{U}_s) ds\right] &= \frac{(\bar{\xi}_t-\bar{U}_t)}{\bar{\rho}}\int_t^T e^{-\int_t^s F^{\Delta}(r)dr}ds \\
 &=  \frac{(\bar{\xi}_t-\bar{U}_t)}{\bar\rho}\int_t^T \frac{\cosh(\sqrt{\Delta}(T-s))}{\cosh(\sqrt{\Delta}(T-t))}ds=  \frac{F^{\Delta}(t)(\bar{\xi}_t-\bar{U}_t)}{\Delta\bar\rho}.
\end{align*}
Differentiation in turn shows that the dynamics of the price adjustment are
\begin{align*}
d(S_t-\scr{D}_t) &=\frac{1}{\Delta \bar\rho} \left(\frac{d}{dt} F^{\Delta}(t)-F^{\Delta}(t)^2\right)(\bar{\xi}_t-\bar{U}_t)dt+\frac{F^{\Delta}(t)}{\Delta\bar{\rho}}q_cd\xi^c_t\\
&=-\frac{1}{\bar\rho}(\bar{\xi}_t-\bar{U}_t)dt+\frac{F^{\Delta}(t)}{\Delta\bar{\rho}}q_cd\xi^c_t\\
&= -\frac{\Delta}{F^{\Delta}(t)}(S_t-\scr{D}_t)dt+\frac{F^{\Delta}(t)}{\Delta\bar{\rho}}q_cd\xi^c_t.
\end{align*}
Price deviations are pulled towards zero, with fluctuations driven by the clients trading targets. Far from maturity (or in the large-liquidity limit), $F^{\Delta}(t) \approx \sqrt{\Delta}$ so that the price correction thus approximately has Ornstein-Uhlenbeck dynamics:
$$
d(S_t-\scr{D}_t) \approx -\sqrt{\Delta}(S_t-\scr{D}_t)dt+\frac{1}{\sqrt{\Delta}}\frac{q_c}{q_c\rho_c+q_d\rho_d}d\xi^c_t
$$
Accordingly, fluctuations in clients' demands have a smaller initial
effect and decay faster if $\Delta$ is large for a liquid open
market (small $\lambda$) as well as for many or for more risk tolerant agents. Close to maturity, price impact tends to zero and the strength $\Delta/F^{\Delta}(t)$ with which prices are pulled towards fundamentals explodes, so that the equilibrium price in the dealer market approaches the exogenous terminal payoff of the risky asset as illustrated in the right panel of Figure~\ref{fig2}.

With randomly fluctuating client demand, not just the expected return but also the volatility of the equilibrium price can change relative to the expected dividend. The above calculations show that the magnitude of this effect is governed by $\approx \sigma_\xi q_c/\sqrt{\Delta}\bar{\rho}$; its sign is in turn determined by the correlation between the expected dividend and the clients' demand. If this correlation is positive, i.e., clients' demand tends to increase when expected fundamentals rise, then illiquidity (caused by an illiquid open market and few or risk-averse dealers) increases volatility. If the correlation is negative, the sign is reversed. The interpretation is that demand pressure and fundamental shocks partially offset in the second case, whereas they magnify price fluctuations in the first case. Similar comparative statics also appear in a Radner equilibrium for a market with an exogenous quadratic deadweight costs on trading~\cite{herdegen.al.19}.
\end{example}

\subsection{The Effects of Segmentation}

We now discuss the effects of segmentation between the dealer and open markets. To this end, we consider how the clients' optimal positions and welfare change when they gain access to the open market. Compared to the discussion in the previous section, this means that their individual trading costs in the open market now are finite; to ease notation, we focus on the case where they vanish just like for the dealers ($\lambda^a=0$). Then, we have $\bar{\eta}=(m_c+m_d)M/\lambda$ and the corresponding parameter determining trading and equilibrium prices in the integrated market is larger than for the segmented market considered in Section~\ref{ss:segmented},
$$
\Delta^{\mathrm{int}}=\frac{1}{(q_c\rho_c+q_d\rho_d)\lambda(1+\frac{1}{(m_c+m_d)M})} > \frac{1}{(q_c\rho_c+q_d\rho_d)\lambda(1+\frac{1}{m_dM})} = \Delta.
$$
Let us now discuss what this implies for our concrete examples. For optimal liquidation, Theorem~\ref{thm:Nash}(i) and (iii) show that the total optimal position the clients take in the dealer and in the open market is 
$$
K_t^{c,\mathrm{int}}=\frac{q_d\rho_d}{q_c\rho_c+q_d\rho_d}\xi^c+q_d\frac{\rho_c-\rho_d}{q_c\rho_c+q_d\rho_d}\bar{U}^{\mathrm{int}}_t,
$$
where
$$
\bar U_t^{\mathrm{int}}=\left(1-\frac{\cosh(\sqrt{\Delta^{\mathrm{int}}}(T-t))}{\cosh(\sqrt{\Delta^{\mathrm{int}}}T)}\right)q_c \xi^c.
$$
Hence, the initial block trade remains unchanged compared to the segmented market. However, the clients
subsequently build up larger positions in the integrated market with parameter $\Delta^{\mathrm{int}}$ rather than $\Delta$. Similarly, it can be shown that also clients with diffusive trading targets share the same fraction of their Brownian shocks with the dealers, but increase the order flow to the open market if they have direct access themselves. In the corresponding equilibrium price dynamics, the initial price impact in the liquidation model disappears faster; for diffusive trading targets, price impact is both smaller and disappears faster.

Finally, let us discuss the effect of segmentation on the clients' welfare as measured by their goal functionals'~\eqref{eq:7}.  In view of Theorem~\ref{thm:Nash}(i,ii,iii), we have
\begin{align*}
J^{c,\mathrm{int}} =& \mathbb{E}\left[\int_0^T \left(\frac{1}{\bar{\rho}}(\bar{U}^{\mathrm{int}}_t-\bar{\xi}_t)K_t^{c,\mathrm{int}}-\frac{1}{2\rho_c}({\xi}_t^c-\bar{U}^{\mathrm{int}}_t-K^{c,\mathrm{int}}_t)^2-\lambda \left(\frac{d}{dt}\bar{U}_t^{\mathrm{int}}\right)^2\right)dt\right].
\end{align*}
Similarly, we can also compute the value of the goal functional in the case the clients do \emph{not} have access to the open market,
\begin{align*}
J^c=& \mathbb{E}\left[\int_0^T \left(\frac{1}{\bar{\rho}}(\bar{U}_t-\bar{\xi}_t)K_t^c-\frac{1}{2\rho_c}(K^c_t-{\xi}_t^c)^2\right)dt\right].
\end{align*}
 For the optimal liquidation example, we have 
 \begin{align*}
 J^{c,\mathrm{int}} =&-(\xi^c)^2\int_0^T \frac{q_c q_d}{\bar{\rho}}\frac{\cosh(\sqrt{\Delta^{\mathrm{int}}}(T-t))}{\cosh(\sqrt{\Delta^{\mathrm{int}}}T)}\left(1+\frac{q_c(\rho_d-\rho_c)}{\bar{\rho}}\frac{\cosh(\sqrt{\Delta^{\mathrm{int}}}(T-t))}{\cosh(\sqrt{\Delta^{\mathrm{int}}}T)}\right)  \\
&\quad\quad\quad\quad+\frac{q_c^2\rho_c}{2\bar{\rho}^2}\left(\frac{\cosh(\sqrt{\Delta^{\mathrm{int}}}(T-t))}{\cosh(\sqrt{\Delta^{\mathrm{int}}}T)}\right)^2+{\lambda q_c^2 \Delta^{\mathrm{int}}} \left(\frac{\sinh(\sqrt{\Delta^{\mathrm{int}}}(T-t))}{\cosh(\sqrt{\Delta^{\mathrm{int}}}T)}\right)^2dt.
\end{align*}

\begin{figure}
  \centering
  \includegraphics[width=.45\linewidth]{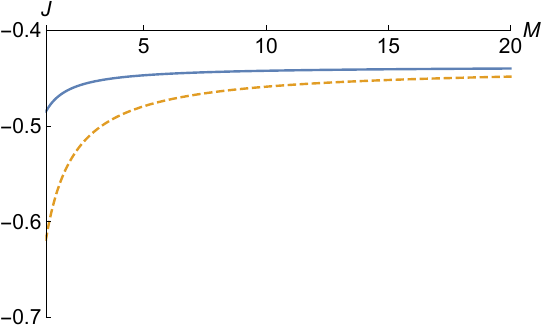}
\caption{Clients' welfare with segmentation (dashed) and in the integrated market (solid) with $M$ dealers and $M$ clients ($q_c=q_d=1/2$). Other parameters are $\lambda=0.1$, $\rho_c=\rho_d=0.1$, $T=1$, and $\xi^c=-1$.}
\label{fig:3}
\end{figure}

An elementary integration in turn gives
 \begin{align*}
 J^{c,\mathrm{int}} =&-\frac{(\xi^c)^2q_c}{\bar \rho}\left(\frac{\tanh(\sqrt{\Delta^{\mathrm{int}}}T)}{\sqrt{\Delta^{\mathrm{int}}}}q_d+\left(\bar \rho-\frac{\rho_c}{2}\right)\frac{q_c}{4\bar \rho\sqrt{\Delta^{\mathrm{int}}}}\frac{\sinh(2\sqrt{\Delta^{\mathrm{int}}}T)   +2\sqrt{\Delta^{\mathrm{int}}}T}{\cosh^2(\sqrt{\Delta^{\mathrm{int}}}T)}\right)  \\
&-(\xi^c)^2\frac{\lambda q_c^2\sqrt{\Delta^{\mathrm{int}}}}{4} \frac{\sinh(2\sqrt{\Delta^{\mathrm{int}}}T) -2\sqrt{\Delta^{\mathrm{int}}}T  }{\cosh^2(\sqrt{\Delta^{\mathrm{int}}}T)}
 \end{align*}
Similarly,
\begin{align*}
J^c&=-\frac{(\xi^c)^2 q_c}{\bar{\rho}}\int_0^T\frac{\cosh(\sqrt{\Delta}(T-t))}{\cosh(\sqrt{\Delta}T)}\left(1-\frac{q_c\rho_c}{2\bar{\rho}}\frac{\cosh(\sqrt{\Delta}(T-t))}{\cosh(\sqrt{\Delta}T)}\right)dt\\
&=-\frac{(\xi^c)^2 q_c}{\bar{\rho}}\left(\frac{\tanh(\sqrt{\Delta}T)}{\sqrt{\Delta}}-\frac{q_c\rho_c}{2\bar{\rho}}\frac{\sinh(2\sqrt{\Delta}T)+2\sqrt{\Delta}T}{4\sqrt{\Delta} \cosh^2(\sqrt{\Delta}T)}\right).
\end{align*}
Now suppose the number clients and dealers is large ($M \to \infty$), whereas the proportions of each group in the total population remain fixed (fixed $m_c,m_d$ and in turn $q_c,q_d,\bar{\rho}$). Then, $\Delta,\Delta^{\mathrm{int}}\to \Delta_\infty= 1/\bar{\rho}\lambda$ and in turn
\begin{align*}
\frac{J^c}{J^{c,\mathrm{int}}} &\to \frac{\tanh(\sqrt{\Delta_\infty}T)-\frac{q_c\rho_c}{2\bar{\rho}}\frac{\sinh(2\sqrt{\Delta_\infty}T)+2\sqrt{\Delta_\infty}T}{4\cosh^2(\sqrt{\Delta_\infty}T)}}{\tanh(\sqrt{\Delta_\infty}T)q_d+\left(\bar \rho-\frac{\rho_c}{2}\right)\frac{q_c}{4\bar \rho}\frac{\sinh(2\sqrt{\Delta_\infty}T)   +2\sqrt{\Delta_\infty}T}{\cosh^2(\sqrt{\Delta_\infty}T)} +\frac{q_c}{4} \frac{\sinh(2\sqrt{\Delta_\infty}T) -2\sqrt{\Delta_\infty}T  }{\cosh^2(\sqrt{\Delta_\infty}T)}}\\
&=\frac{\tanh(\sqrt{\Delta_\infty}T)-\frac{q_c\rho_c}{2\bar{\rho}}\frac{\sinh(2\sqrt{\Delta_\infty}T)+2\sqrt{\Delta_\infty}T}{4\cosh^2(\sqrt{\Delta_\infty}T)}}{\tanh(\sqrt{\Delta_\infty}T)q_d-\frac{q_c\rho_c}{2\bar{\rho}}\frac{\sinh(2\sqrt{\Delta_\infty}T)+2\sqrt{\Delta_\infty}T}{4\cosh^2(\sqrt{\Delta_\infty}T)}+\frac{q_c}{2}\frac{\sinh(2\sqrt{\Delta_\infty}T)}{\cosh^2(\sqrt{\Delta_\infty}T)}}=1, \quad \mbox{as $M\to \infty$.}
\end{align*}
(Here, we have used $\sinh(2x)=2\sinh(x)\cosh(x)$, $\tanh(x)=\sinh(x)/\cosh(x)$, and $q_c+q_d=1$ in the last step.) This shows that, in the competitive limit of many small clients and dealers,  the welfare loss due to segmentation  disappears, as illustrated in Figure~\ref{fig:3}. Maybe surprisingly, this suggests that market segmentation can be close to optimal as long as none of the agents has substantial market power.

\appendix

\section{Proof of Theorem~\ref{thm:Nash}}\label{s:proof}

\begin{proof}[Proof of Theorem~\ref{thm:Nash}]
  Let $S \in \cH^2$, $(K^a,u^a) \in \cS^2\times \cL^2$, $a \in \cA$,
  be \emph{any} equilibrium. It is easy to see that each agent's
  target functional $J^a(K,u;\bar{u}^{-a},S)$ is strictly concave in
  $(K,u) \in \cS^2 \times \cL^2$, and so its maximizer $(K^a,u^a)$ is
  uniquely determined by the first-order condition that all
  directional derivatives vanish. In particular, for any $K \in \cS^2$ we must have
  \begin{align}
    \label{eq:17}
     0 = \nabla_K J^a(K^a,u^a;\bar{u}^{-a},S) = \E \left[\int_0^T K_t dA_t +
    \int_0^T \frac{1}{\rho^a}(\xi^a_t-K^a_t-U^a_t) K_t dt\right].
  \end{align}
  This can only hold true if $A$ is absolutely continuous with density 
  \begin{align}
    \label{eq:18}
      \frac{dA_t}{dt} = \mu_t =\frac{1}{\rho^a}\left(U^a_t+K^a_t-\xi^a_t\right), \quad
    0 \leq t \leq T, \quad a \in \mathcal{A}.
  \end{align}
   Solving for $K^a$, we find 
   \begin{align}
     \label{eq:26}
        K^a = \xi^a-U^a+\rho^a \mu.
   \end{align}
   Hence, after aggregating over $a \in \cA$, the
   market clearing condition $K^N + \sum_{a \in \cA} m(a) K^a=0$ implies $- K^N = \bar{\xi}-\bar{U}+\bar{\rho} \mu$ and in turn~\eqref{eq:13}. 

   The first-order conditions for individual optimality also require that, for any $u \in \cL^2$,
   \begin{align*}
      0 = & \nabla_u
     J(K^a,u^a;\bar{u}^{-a},S)\\=&-\E\left[\int_0^Tu_t\left(\lambda
     \bar{u}_t^{-a}+(2m(a)\lambda+\lambda^a)u^a_t+\E_t\left[\int_t^T
     \frac{1}{\rho^a}\left(K^a_s+U^a_s-\xi^a_s\right)ds
\right]\right)dt\right].
   \end{align*}
   Since this equality needs to hold for any perturbation $u  \in \cL^2$, it is equivalent to
   \begin{align}
     \label{eq:21}
     \lambda
     \bar{u}_t^{-a}+(2m(a)\lambda+\lambda^a)u^a_t+\E_t\left[\int_t^T
     \frac{1}{\rho^a}\left(K^a_s+U^a_s-\xi^a_s\right)ds
     \right]=0, \quad
     0 \leq t \leq T.
   \end{align}
   Recalling that $\eta^a=1/(m(a)\lambda+\lambda^a)$, the definition of $\bar{u}$ and~\eqref{eq:18}, this implies
   \begin{align}
     \label{eq:22}
      u_t^a = -\eta^a\left(\lambda \bar{u}_t+\E_t\left[\int_t^T\mu_sds\right]\right) , \quad
    0 \leq t \leq T.
   \end{align}
    Now, we aggregate over $a \in \cA$ to obtain
    \begin{align*}
       \bar{u}_t = -\bar{\eta} \left(\lambda\bar{u}_t+\E_t\left[\int_t^T\mu_sds\right]\right) , \quad
    0 \leq t \leq T.
    \end{align*}
    Solving for $\bar{u}$ and using the already established
    relation~\eqref{eq:13} along with 
    $\Delta = \bar{\eta}/(\bar{\rho}(1+\bar{\eta}\lambda))$, we find
     \begin{align}
       \label{eq:24}
       \bar{u}_t =-\Delta
       \E_t\left[\int_t^T\left(\bar{U}_s-K^N_s-\bar{\xi}_s\right)ds\right], \quad
       0 \leq t \leq T.
     \end{align}
     Therefore, the pair $(u,U):=(\bar{u},\bar{U})$ solves the linear
     FBSDE~\eqref{eq:25} for $X=K^N+\bar{\xi}$. As this FBSDE is
     uniquely solved by $(\mathbf{u}^\Delta(X),\mathbf{U}^\Delta(X))$
     from~\eqref{eq:30}, this pins down $\bar{U}$ as stated
     in~\eqref{eq:8} with density $\bar{u}$ as given in~\eqref{eq:24}.
     From~\eqref{eq:22}, we infer that each individual agent's
     strategy $u^a$, $a \in \cA$, is the multiple
     $\eta^a=1/(m(a)\lambda+\lambda^a)$ of the same universal
     process. Since $\bar{u}=\sum_{a \in \cA} m(a)u^a$, this implies
     $u^a = \frac{\eta^a}{\bar{\eta}} \bar{u}$ and, thus,
     $U^a=\frac{\eta^a}{\bar{\eta}} \bar{U}$ as claimed. Moreover, in light of~\eqref{eq:26}
     and~\eqref{eq:13}, our observation that
     $U^a=\frac{\eta^a}{\bar{\eta}} \bar{U}$ allows us to write each
     agent's position in the dealer market in the form~\eqref{eq:15}.  As
     a consequence, the only candidate for an equilibrium is the one
     described in the present theorem.
     
       We now show that strategies and price pinpointed above
  form indeed an equilibrium. For this purpose, define $ \bar U$ by
  \eqref{eq:8} and $\bar u=\frac{d}{dt} \bar U$ so that $(\bar u, \bar
  U)$ solves the FBSDE \eqref{eq:25}. For each $a\in \cA$ the
  candidate equilibrium position is then $K^a$ given by \eqref{eq:15}
  and the candidate equilibrium trading rate in the open market is
  $u^a=\frac{\eta_a}{\bar \eta} \bar u$. The candidate equilibrium
  price is $S$ with risk premium $\mu$ given by \eqref{eq:13}. 
 It is readily checked that the positions $K^a$, $a \in \cA$, ensure
 market clearing. Moreover, the first order condition in $K^a$ is 
  \begin{align*}
     0 = \E \left[\int_0^T K_t \mu_t dt +
    \int_0^T \frac{1}{\rho^a}(\xi^a_t-K^a_t-U^a_t) K_t dt\right], \quad \mbox{ for all $K \in \mathcal{S}^2$},
  \end{align*}
which is clearly satisfied by the candidate strategy $K^a$ by definition of the corresponding drift rate $\mu$. Using the
first-order condition satisfied by $K^a$, the first-order condition for $u^a$ then amounts to
\eqref{eq:22} by the same reasoning as above. The latter condition is
readily verified since $(\bar u, \bar U)$ were chosen as the solution
to the FBSDE \eqref{eq:25}. Concavity of the goal functional ensures
sufficiency of the first-order conditions, and so $(K^a,u^a)$ maximizes
$J^a$ as required for an equilibrium.
\end{proof}

\section{Proofs for Section~\ref{ss:asym}}\label{a:asymptotics}
By the definition of the equilibrium, the optimization criterion of each dealer is
\begin{align*}
  J^a(K^a,u^a;\bar{u},S) := \E\left[\int_0^T K^a_t dA_t - \int_0^T \lambda \bar u_t u^a_t- \frac{1}{2\rho_d}(K^a_t+U^a_t)^2dt\right].
\end{align*}
Similarly as in the Proof of Theorem \ref{thm:Nash}, after aggregating over all the agents, the first-order optimality condition for $u^a$ is 
$$       \bar{u}_t =-\frac{M}{\lambda\rho_d(M+1)}
       \E_t\left[\int_t^T\left(\bar{U}_s-K^N_s\right)ds\right], \quad
       0 \leq t \leq T.$$
Therefore $\bar u$ is an optimizer of the auxiliary minimization problem 
\begin{align}\label{eq:aux}
\min_{u}\E\left [\int_0^T \lambda  u^2_t +\frac{1}{2}\frac{M}{\rho_d(M+1)}\left(K^N_t-\int_0^t u_s ds\right)^2dt\right].
\end{align}

\begin{proof}[Proof of Proposition~\ref{prop:priceasym}]
Set
$M_t=\E_t\left[\int_0^T \mu_s ds\right]$
so that 
$d(\scr{D}_t-S_t)=dM_t-\mu_t dt.$ By the $\varepsilon$-Young inequality, the Cauchy-Schwarz inequality and Doob's maximal inequality, for all predictable $H$ bounded by $1$, we have
\begin{align*}
\E\left[\sup_{0\leq t\leq T}\left|\int_0^t H_s d(\scr{D}_s-S_s)\right|^2\right]\leq 8\E\left[ M^2_T\right]+2T\E\left[\int_0^T  |\mu_s|^2 ds\right].
\end{align*}
Together with the equality $M_T=\int_0^T \mu_s ds$ and a second application of the Cauchy-Schwarz inequality, it follows that
\begin{align*}
\E\left[\sup_{0\leq t\leq T}\left|\int_0^t H_s d(\scr{D}_s-S_s)\right|^2\right]&\leq 8\E\left[ \left(\int_0^T \mu_s ds\right)^2\right]+2T\E\left[\int_0^T  |\mu_s|^2 ds\right]\\
&\leq10T\E\left[\int_0^T  |\mu_s|^2 ds\right]\leq \frac{10T}{\bar{\rho}^2}\E\left[\int_0^T  \left(K^N_s-\bar{U}_s\right)^2 ds\right].
\end{align*}
To prove the proposition it is therefore sufficient to show that the last term converges to $0$ as $\lambda\to 0$. 
The set $\{U=\int_0^. u_s ds : u \in \scr{L}^2 \}$ is dense in
$\scr{L}^2$, so that there exists a sequence $u^n\in \scr{L}^2$ such that 
$$\E\left[\int_0^T  \left(\int_0^tu^n_sds-K^N_t\right)^2dt\right]\leq\frac{1}{n}, \quad n=1,2,\ldots$$
Due to the minimality condition \eqref{eq:aux}, we have 
\begin{align*}
\E\left[\int_0^T \frac{1}{2}\frac{M}{\rho_d(M+1)}(\bar U_t-K^N_t)^2dt\right] &\leq \E\left[\int_0^T \left(\frac{\lambda}{2}(u^n_t)^2+\frac{1}{2}\frac{M}{\rho_d(M+1)}\left(\int_0^tu^n_sds-K^N_t\right)^2\right)dt\right]\\
&\ \leq \frac{\lambda}{2}\E\left[\int_0^T (u^n_t)^2dt\right]+\frac{1}{2\rho_dn}.
\end{align*}
Thus, $\E[\int_0^T  (\bar U_t-{K}^N_t)^2dt]
\to 0$ as $\lambda \to 0$, verifying the first convergence asserted in Proposition~\ref{prop:priceasym}.

To also establish the second convergence result, we apply the inequalities of Burkholder-Davis-Gundy and H\"older to obtain 
\begin{align*}
\E\left[\left|\int_0^TK^N_t (dS_t -d\scr{D}_t)\right|\right]&\leq   \E\left[\sup_{0 \leq t \leq T} \left|\int_0^t K^N_s dM_s\right|\right]+\ E\left[\sup_{0 \leq t \leq T} \left|\int_0^t K^N_s \mu_sds\right|\right]\\
&\leq   C\left(\E\left[\left(\int_0^T ({K}^N)_s^2 d\langle M\rangle_s\right)^{1/2}\right]+ \E \left[\int_0^T |K^N_s| |\bar U_s-K^N_s|ds\right]\right)\\
&\leq C\left(\E\left[\sup_{s\in[0,T]}|K^N_s|\left(\langle M\rangle _T^{1/2}+\int_0^T  |\bar U_s-K_s^N|ds\right)\right]\right)\\
&\leq C \E\left[\sup_{s\in[0,T]}|K^N_s|^{2}\right]^{1/2}\E\left[\int_0^T  |\bar U_s-K^N_s|^2ds\right]^{1/2}.
\end{align*}
Here, $C>0$ is a constant that might change from line to line but does not depend on $\lambda$. $L^1$ convergence now follows, since we have already verified above that the last term converges to $0$ as $\lambda\to 0$.
\end{proof}

\begin{proof}[Proof of Lemma~\ref{lem:costsasym1}]
By \eqref{eq:13}, \eqref{eq:24}, and the integration by parts formula (using $K^N_0=\bar u_T$), we have
\begin{align}\label{eq:tcforabs}\int_0^T {K^N_t} d\scr{D}_t -\int_0^T {K^N_t}dS_t =-\lambda\frac{M+1}{M} \int_0^T {K^N_t} d\bar u_t = \lambda\frac{M+1}{M} \int_0^T \mu^N_t \bar u_t dt.
\end{align}
To establish~\eqref{eq:limit1}, it therefore suffices to show that
\begin{equation}\label{eq:suff1}
\int_0^T |\mu^N_t-\bar u_t|dt = o(1) \quad \mbox{in $L^1$~as $\lambda \to 0$.}
\end{equation}
Integration by parts and \eqref{eq:30} give
\begin{align}
\mu_t^N-\bar u_t&=\mu_t^N+F^{\Delta}(t)\bar U_t -\frac{{ \Delta}}{ \cosh(\sqrt{\Delta}(T-t))} \E_t\left[\int_t^T{ \cosh(\sqrt{\Delta}(T-s))}K^N_sds\right] \notag\\
&=\mu_t^N-F^{\Delta}(t)({K^N_t}-\bar U_t)-\sqrt{ \Delta}\E_t\left[\int_t^T\frac{\sinh(\sqrt{\Delta}(T-s))}{ \cosh(\sqrt{\Delta}(T-t))} \mu^N_sds\right]\\
&=\mu_t^N\left(1-\sqrt{ \Delta}\int_t^T\frac{\sinh(\sqrt{\Delta}(T-s))}{ \cosh(\sqrt{\Delta}(T-t))} ds\right)\\
&\quad-F^{\Delta}(t)({K^N_t}-\bar U_t)+\sqrt{ \Delta}\E_t\left[\int_t^T\frac{\sinh(\sqrt{\Delta}(T-s))}{ \cosh(\sqrt{\Delta}(T-t))} (\mu_t^N-\mu^N_s)ds\right]. \label{eq:ibp}
\end{align}
Now, note that 
$$\frac{\sqrt{ \Delta}\int_t^T{ {\sinh(\sqrt{\Delta}(T-s))}}ds}{ \cosh(\sqrt{\Delta}(T-t))}=1-\frac{1}{ \cosh(\sqrt{\Delta}(T-t))}.$$
Together with~\eqref{eq:ibp}, it follows that $K^N-\bar U$ satisfies the linear ODE
\begin{equation}\label{eq:odediff}
\frac{d({K^N_t}-\bar U_t)}{dt}= \mu^N_t-\bar u_t=-F^{\Delta}(t)({K^N_t}-\bar U_t)+w^\Delta(t)+\frac{\mu^N_t}{ \cosh(\sqrt{\Delta}(T-t))},
\end{equation}
where
$$
w^\Delta(t)=\E_t\left[\int_t^T\frac{\sqrt{ \Delta} \sinh(\sqrt{\Delta}(T-s))}{ \cosh(\sqrt{\Delta}(T-t))} (\mu^N_t-\mu^N_s)ds\right].
$$
Since $K^N_0=\bar U_0=0$ and by definition of the function $F^\Delta$ the explicit solution of~\eqref{eq:odediff} is
\begin{align*}
{K^N_t}-\bar U_t &= \int_0^t e^{-\int_s^t F^\Delta(u)du}\left(w^\Delta_s+\frac{\mu^{N}_s}{ \cosh(\sqrt{\Delta}(T-s))}\right) ds\\
&=\int_0^t \frac{\cosh(\sqrt{\Delta}(T-t))}{\cosh(\sqrt{\Delta}(T-s))}\left(w^\Delta_s+\frac{\mu^N_s}{ \cosh(\sqrt{\Delta}(T-s))}\right) ds.
\end{align*}
Together with~\eqref{eq:odediff}, it follows that
\begin{align}
\int_0^T |\mu^N_t-\bar u_t|dt&\leq\sup_{u \in [0,T]}|w^\Delta_u|\int_0^T \left(1+\sqrt{\Delta}\int_0^t\frac{\sinh(\sqrt{\Delta}(T-t))}{\cosh(\sqrt{\Delta}(T-s))}ds\right)dt \notag\\
&\quad+\sup_{u\in[0,T]}|\mu^{N}_u|\int_0^T \left(\frac{1}{\cosh(\sqrt{\Delta}(T-t))}+\sqrt{\Delta}\int_0^t \frac{\sinh(\sqrt{\Delta}(T-t))}{\cosh^2 (\sqrt{\Delta}(T-s))}ds\right)dt\\
&\leq\sup_{u \in [0,T]}|w^\Delta_u |\int_0^T \left(1+{\sqrt{\Delta}}\int_0^te^{-\sqrt{\Delta}(t-s)}ds\right)dt\\
&\quad+\sup_{u\in[0,T]}|\mu^{N}_u|\left(\int_0^T 2 e^{-\sqrt{\Delta}t}dt+\int_0^T \sqrt{\Delta}\int_s^T \frac{\sinh(\sqrt{\Delta}(T-t))}{\cosh^2 (\sqrt{\Delta}(T-s))}dt ds\right)\\
&\leq 2T\sup_{u\in [0,T]}|w^\Delta_u |+\frac{4}{\sqrt{\Delta}} \sup_{u\in[0,T]}|\mu^{N}_u|. \label{eq:sufficient2}
\end{align}
Write 
$$\omega(\d)=\sup_{t,s \in [0,T], |t-s|\leq \d} \Big|\mu^N(s)-\mu^{N}(t)\Big|$$ 
for the modulus of continuity of $\mu^N$. Since $t \mapsto \mu_t^N$ is continuous on the compact set $[0,T]$,
\begin{equation}\label{eq:convmod}
 \omega(\d)\to 0, \quad \mbox{a.s.~as }\d\to 0.
 \end{equation}
The definitions of $w^{\Delta}$ and $\omega$ and a change of variables yield the following estimate:
\begin{align}\label{eq:estimatew}
|w^\Delta(t)|\leq \frac{1}{2} \E_t\left[\int_t^T\sqrt{ \Delta} e^{-\sqrt{\Delta}(s-t)}\omega(s-t)ds\right]\leq \frac{1}{2} \E_t\left[\int_0^\infty e^{-u}\omega\left(\tfrac{u}{\sqrt{\Delta}}\right)du\right] :=M^\Delta_t.
\end{align}
Here, $(M^{\Delta}_t)_{t \in [0,T]}$ is a martingale for each $\Delta>0$ since $|\omega(\d)|\leq 2\sup_{s \in [0,T]} |\mu_s^N|$ is integrable by assumption. Also note that by definition of the modulus of continuity $\omega$, the mapping $\Delta \mapsto M^{\Delta}_t$ is decreasing for each $t$. Define 
$$
M^*:=\lim_{\Delta \to \infty }\sup_{t \in [0,T]}M^\Delta_t\geq 0.
$$ 
Fix $\e>0$. Then, by the monotonicity in $\Delta$, we have 
$$\P\left[M^*\geq \e\right]\leq \lim_{\Delta \to \infty}\P\left[\sup_{t\in[0,T]} M^\Delta_t\geq \e\right]\leq  \lim_{\Delta \to \infty}\frac{\E[M^\Delta_T]}{\e}=0.$$
Here, the last equality is a consequence of \eqref{eq:convmod}, another application of the monotone convergence theorem and the integrability of right-hand side in \eqref{eq:estimatew}. As a result,
$$0\leq \limsup_{\Delta\to \infty}\sup_{t\in[0,T]}|w^\Delta(t)|\leq M^*=0\mbox{ a.s.}$$
In view of \eqref{eq:sufficient2}, it follows that the asserted convergence~\eqref{eq:suff1} holds in the almost-sure sense.

To show convergence in $L^1$, it therefore suffices to establish uniform integrability of \eqref{eq:tcforabs}. By~\eqref{eq:sufficient2} and~\eqref{eq:estimatew},
\begin{align*}\int_0^T|(\mu_t^N)^2-\bar u_t\mu^N_t| dt \leq \sup_{s\in[0,T]}|\mu^N_s|\int_0^T |\mu_t^N-\bar u_t|dt
&\leq T\sup_{s\in [0,T]}|w^\Delta_s |^2+(T+\frac{4}{\sqrt{\Delta}}) \sup_{s\in[0,T]}|\mu^{N}_s|^2\\
&\leq T\sup_{s\in [0,T]}|M^\Delta_s |^2+(T+\frac{4}{\sqrt{\Delta}}) \sup_{s\in[0,T]}|\mu^{N}_s|^2.
\end{align*}
Observe that the right-hand side is decreasing in $\Delta$, and integrable for, e.g., $\Delta=1$ by Doob's maximal inequality: 
$$\E\left[T\sup_{t\in [0,T]}|M^1_t |^2+(T+4) \sup_{t\in[0,T]}|\mu^N_t|^2\right]\leq \E\left[2T|M^1_T |^2+(T+4) \sup_{t\in[0,T]}|\mu^N_t|^2\right]<\infty.$$
Therefore, the family,  $\{\int_0^T|(\mu_t^N)^2-u_t\mu^N_t| ds,\Delta\geq 1\}$  is uniformly integrable. Since
$$\int_0^T|(\mu_t^N)^2-u_t\mu^N_t| ds \leq \sup_{s\in[0,T]}|\mu^N_s|\int_0^T |\mu_s^N-u_s|ds,$$
 this implies that the almost sure convergence we have established for~\eqref{eq:suff1} also holds in $L^1$.
\end{proof}

\begin{proof}[Proof of Lemma~\ref{prop.costsemi}]
Similarly to \eqref{eq:tcforabs}, the liquidity costs of the clients can be written as
\begin{align}\label{eq:liqcosts}
\int_0^T K^N_t d\scr{D}_t - \int_0^T K^N_t dS_t& = -\lambda\frac{M+1}{M} \int_0^T K^N_t d\bar u_t \notag\\
&= \lambda\frac{M+1}{M} \langle \bar u, K^N \rangle_T+ \lambda\frac{M+1}{M} \int_0^T\bar u_t (\mu^N_tdt+\sigma^N_tdW_t),
\end{align}
and $\bar u_t={\bar K^N}_t-F^\Delta(t) \bar U_t$, where $\bar K^N$ is defined as $\bar K^N_t= \E_t\left[\int_t^T k^\Delta(t,s)K^N_sds\right] $. This implies that the covariation of $\bar u$ is the same as the one of ${\bar K^N}$. Note that by definition of $\bar K^N$, 
\begin{align}\label{eq:1}
\cosh(\sqrt{\Delta}(T-t))\bar K^N_t-\Delta\int_0^t\cosh(\sqrt{\Delta}(T-s))  K^N_sds=\Delta \E_t\left[\int_0^T\cosh(\sqrt{\Delta}(T-s)) K^N_s ds\right]
\end{align}
is a square-integrable martingale. By the martingale representation theorem, it therefore can be written as a stochastic integral with respect to the Brownian motion generating the underlying filtration. The integrand in this representation can be computed using the Clark-Ocone formula. Indeed, setting 
$$\Phi=\Delta \int_0^T\cosh(\sqrt{\Delta}(T-s))K^N_s ds$$
and using the Malliavin differentiability of $\Phi$ that we prove below, the Clark-Ocone formula~\cite[Proposition 1.3.14]{nualart2006malliavin} yields
$$ \Phi=\E\left[\Phi\right]+\int_0^T \E_t\left[D_t \Phi \right] dW_t.$$
By inserting this into \eqref{eq:1} and integrating by parts, we in turn obtain 
\begin{align}\label{eq:covaru}
\langle \bar u,K^N\rangle_T=\int_0^T\frac{\E_t\left[D_t\Phi\right]}{\cosh(\sqrt{\Delta}(T-t))}\sigma^N_t dt. 
\end{align}

We now show that we indeed have $\Phi \in \mathbb{D}^{1,2}$, so that the Clark-Ocone formula can be applied. Given our assumption on the square-integrability of the supremum of its Malliavin derivative, $K^N\in \L^{1,2,f}$, cf.~\cite[p.~45]{nualart2006malliavin}. Thus, by \cite[p.~45]{nualart2006malliavin}, $\Phi$ is Malliavin differentiable and it follows from the product rule that 
\begin{align}\label{eq:deriv}
&D_t \Phi=\Delta\int_t^T{\cosh(\sqrt{\Delta}(T-s))}D_tK^N_s ds.
\end{align}
We now expand $D_t\Phi$ and $\E_t[D_t\Phi]$ for $\lambda \to 0$ or, equivalently, $\Delta \to \infty$. First note that by~\eqref{eq:deriv} and the definition of the hyperbolic cosine,
\begin{align}\label{eq:cosh}
\frac{\sqrt{\Delta}^{-1}D_t\Phi}{\cosh(\sqrt{\Delta}(T-t))}&=\sqrt{\Delta} \int_t^T\frac{e^{\sqrt{\Delta}(t-s)}+e^{-\sqrt{\Delta}(2T-(s+t))}}{1+e^{-2\sqrt{\Delta}(T-t)}}D_tK^N_sds. 
\end{align}
By continuity of $s\mapsto D_t(K^N_s)$ on $[t,T]$, some elementary integrations show that the above expression converges to $D_tK^N_t$ as $\Delta\to \infty$. In view of~\cite[Proposition 1.3.8]{nualart2006malliavin}, we have $D_t(\int_0^t \sigma^N_s dW_s)=\sigma^N_t$. Moreover, $D_t(\int_0^t \mu^N_s ds)=0$, so that
$$
\frac{\sqrt{\Delta}^{-1}D_t \Phi}{\cosh(\sqrt{\Delta}(T-t))}\to \sigma_t^N, \quad \mbox{$P$-a.s.~as $\Delta\to \infty$}.
$$
Next, observe that for every $t\in [0,T]$, it follows from~\eqref{eq:cosh} that
\begin{align}
\sup_{\Delta>1} \left|\frac{\sqrt{\Delta}^{-1}D_t \Phi}{\cosh(\sqrt{\Delta}(T-t))}\right| &\leq  2\sup_{s\in [t,T]}| D_tK^N_s|   \sup_{\Delta>1}\left\{\sqrt{\Delta}\int_t^Te^{\sqrt{\Delta}(t-s)}ds\right\}\\
& \leq 2 \sup_{ t\leq s\leq T}| D_tK^N_s|.\label{eq:bound1}
\end{align}
Since the right-hand side is integrable by assumption, the dominated convergence theorem in turn shows 
$$\E_t \left[\frac{\sqrt{\Delta}^{-1}D_t\Phi}{\cosh(\sqrt{\Delta}(T-t))}\right]\to \sigma_t^N, \quad \mbox{$dP\times dt$-a.s.~as $\Delta\to \infty$.}$$
We now show that this expansion of $D_t\Phi$ is inherited by its conditional expectation and in turn the covariation~\eqref{eq:covaru}. To this end, we first use \eqref{eq:bound1} and Young's inequality to obtain that 
\begin{align*}
\sup_{\Delta>1}  |\sigma^N_t|\left|\E_t \left[\frac{\sqrt{\Delta}^{-1}D_t\Phi}{\cosh(\sqrt{\Delta}(T-t))}\right]\right|&\leq \frac{1}{3} \sup_{t\in[0,T]} |\sigma^N_t|^3+\frac{2}{3}\sup_{\Delta>1} \left|\E_t \left[\frac{\sqrt{\Delta}^{-1}D_t\Phi}{\cosh(\sqrt{\Delta}(T-t))}\right]\right|^{3/2}\\
&\leq \frac{1}{3} \sup_{t\in[0,T]} |\sigma^N_t|^3+\frac{\sqrt{32}}{3} \E_t \left[ \sup_{t\leq s\leq T}| D_tK^N_s|^{3/2}\right].
\end{align*}
Jensen's inequality and the integrability assumption for the supremum of the Malliavin derivative of ${K^N}$ yield 
$$\E\left[\int_0^T \E_t \left[ \textstyle{\sup_{t\leq s\leq T}| D_tK^N_s|^{3/2}}\right]^{4/3}dt\right]\leq   \E\left[\int_0^T \textstyle{\sup_{t\leq s\leq T}| D_tK^N_s|^{2}}dt\right]<\infty.$$
Moreover, $(\sup_{t\in[0,T]} |\sigma^N_t|^3)^{4/3}=\sup_{t\in[0,T]} |\sigma^N_t|^4$ is also integrable by assumption. Together, these two estimates show that
\begin{align}\label{eq:bound43}
\E\left[\int_0^T \left(\sup_{\Delta>1}  |\sigma^N_t|\left|\E_t \left[\frac{\sqrt{\Delta}^{-1}D_t \Phi}{\cosh(\sqrt{\Delta}(T-t))}\right]\right|\right)^{4/3}dt\right]<\infty.
\end{align}  
Since the term inside this expectation is finite, the dominated convergence theorem implies that, as $\lambda \to 0$ and in turn $\Delta \to \infty$,
$$\sqrt{\Delta}^{-1}\langle \bar u,K^N\rangle_T=\int_0^T\frac{\sqrt{\Delta}^{-1}\E_t[D_t\Phi]}{\cosh(\sqrt{\Delta}(T-t))}\sigma^N_t dt\to \int_0^T \left(\sigma^N_t\right)^2dt,\quad  \mbox{$P$-a.s.}$$
Finally, \eqref{eq:bound43} also shows that the $4/3$-th moment of $\int_0^T\frac{\sqrt{\Delta}^{-1}E_t[D_t\Phi]}{\cosh(\sqrt{\Delta}(T-t))}\sigma^N_t dt$ is bounded, uniformly for all $\Delta >1$. Therefore, this family indexed by $\Delta>1$ is uniformly integrable and the almost sure convergence for $\Delta \to \infty$ also holds in $L^1$. 

\medskip

To complete the proof, we now show that the other terms in~\eqref{eq:liqcosts} do not contribute at the leading order $O(\sqrt{\lambda})$, that is,
$$\lambda \int_0^T \bar u_t( \mu^N_tdt+\sigma^N_tdW_t)=o(\sqrt{\lambda}), \quad \mbox{in $L^1$ as $\lambda \to 0$.}$$
Since $\Delta=\frac{M}{\lambda\rho_d(M+1)}$, this is implied by a bound for ${\Delta}^{-1/2} u^\mathcal{K}$. To this end, observe that the inequalities of Jensen and Burkholder-Davis-Gundy show
\begin{align*}
\E\left[(K^N_t-K^N_s)^4\right] &\leq C\left( \E\left[\left(\int_s^t \mu^N_r dr\right)^4\right]+ \E\left[\left(\int_s^t \sigma^N_r dW_r\right)^4\right]\right)\\
&\leq C' \E\left[\sup_{0\leq u\leq T}\left\{|\mu_u^N|^4+|\sigma_u^N|^4\right\}\right](t-s)^2,
\end{align*}
for some constants $C,C'>0$ that might only depend on $T$. For $\a<\frac{1}{4}$, write $R_\a$ for the modulus of $\a$-H\"older continuity of $K^N$. This quantity is well defined and satisfies $E[R_\a^4]<\infty$ by~\cite[Theorem 3.1]{friz2014course}.\footnote{This theorem requires an additional assumption on the iterated integral of the process  $K^N$. However, a careful inspection of the proof reveals that this extra assumption is only needed to establish additional path regularity of the iterated integral and not for the path regularity of the process $K^N$ itself.} As $K^N \in \mathcal{H}^2$ by assumption, we can define the square-integrable random variable 
$$
M^\a:= \sup_{s\in[0,T]}\E_s[R_\a]\left(1+ \int_0^\infty {2 e^{-u}} |u|^\a ds\right)+\sup_{s\in[0,T]}\left|K^N_s\right| <\infty.
$$
Then, we can estimate
\begin{align*}
\left|{\Delta}^{-1/2}\bar K^N_t-K^N_t\right|&\leq \E_t\left[R_\a \int_t^T \Delta^{-1/2} k^\Delta(t,s)|t-s|^\a ds\right]\\
&\qquad+ \left|1-\int_t^T\Delta^{-1/2} k^\Delta(t,s)ds\right|\left|K^N_t\right|\\
&\leq  \E_t[R_\a]{\Delta ^{(1-\a)/2}} \int_t^{T} \frac{2 e^{\sqrt{\Delta} (T-s)}}{e^{\sqrt{\Delta} (T-t)}} |\sqrt{\Delta} (t-s)|^\a ds\\
&\qquad+\left|1-\tanh\left(\sqrt{\Delta}(T-t)\right)\right| |K^N_t|\\
&\leq  \E_t[R_\a]{\Delta ^{-\a/2}} \int_0^\infty {2 e^{-u}} |u|^\a ds+2e^{-2\sqrt{\Delta}(T-t)}\left|K^N_t\right|\\
&\leq  \left( \Delta ^{-\a/2}+2e^{-2\sqrt{\Delta}(T-t)}\right)M^\a=C_{t,T,\a,\lambda}.
\end{align*}
Together with the formulas for $\bar u,\bar U$ and the definition of the function $F^\Delta$ and \eqref{eq:30}, this estimate yields
\begin{align}
\left|{\Delta}^{-1/2}\bar u_t\right|&=\left|-{\Delta}^{-1/2}F^\Delta(t)\bar U_t+{\Delta}^{-1/2}\bar K_t^N\right| \notag\\
&=\left|-\sqrt{\Delta}\int_0^t\frac{\sinh(\sqrt{\Delta}(T-t))}{\cosh(\sqrt{\Delta}(T-s))}{\Delta}^{-1/2}\bar K^N_sds+{\Delta}^{-1/2}\bar K^N_t\right| \notag\\
&\leq\left|\sqrt{\Delta}\int_0^t\frac{\sinh(\sqrt{\Delta}(T-t))}{\cosh(\sqrt{\Delta}(T-s))}K^N_sds-K^N_t\right|\\
&\qquad+C_{t,T,\a,\lambda}+\sqrt{\Delta}\int_0^t\frac{\sinh(\sqrt{\Delta}(T-t))}{\cosh(\sqrt{\Delta}(T-s))}C_{s,T,\a,\lambda}ds \notag\\
&\leq \sqrt{\Delta}\int_0^t\frac{ \sinh(\sqrt{\Delta}(T-t))}{\cosh(\sqrt{\Delta}(T-s))}\left(R_\a|t-s|^\a+C_{s,T,\a,\lambda}\right) ds \notag\\
&\quad+C_{t,T,\a,\lambda}+\left|1-\sqrt{\Delta}\int_0^t\frac{\sinh(\sqrt{\Delta}(T-t))}{\cosh(\sqrt{\Delta}(T-s))} ds\right|\left|K^N_t\right| \notag\\
&\leq C_{t,T,\a,\lambda}+\Delta^{-\a/2 }R_\a\int_0^\infty e^{-u}  |u|^\a du+\sqrt{\Delta}\int_0^t e^{-\sqrt{\Delta}(t-s)}C_{s,T,\a,\lambda} ds \notag\\
&\quad+\left|1-\sqrt{\Delta}\int_0^t\frac{\sinh(\sqrt{\Delta}(T-t))}{\cosh(\sqrt{\Delta}(T-s))} ds\right|\left|K^N_t\right|. \label{eq:ests}
\end{align}
Recall the addition formula $\arctan\left(x\right)-\arctan\left(y\right)=\arctan\left(\frac{x-y}{1+xy}\right)$ for $x,y\geq 0$ and observe that $|\arctan(x)|\leq |x|$. As a consequence:
\begin{align*}
\sqrt{\Delta}\int_0^t\frac{\sinh(\sqrt{\Delta}(T-t))}{\cosh(\sqrt{\Delta}(T-s))} ds&={\sinh(\sqrt{\Delta}(T-t))}\left(\arctan\left(\sinh(\sqrt{\Delta}T)\right)-\arctan\left(\sinh(\sqrt{\Delta}(T-t))\right)\right)\\
&={\sinh(\sqrt{\Delta}(T-t))}\arctan\left(\frac{\sinh(\sqrt{\Delta}T)-\sinh(\sqrt{\Delta}(T-t))}{1+\sinh(\sqrt{\Delta}T)\sinh(\sqrt{\Delta}(T-t))}\right)\\
&\leq {\sinh(\sqrt{\Delta}(T-t))}\arctan\left(\frac{\sinh(\sqrt{\Delta}T)}{1+\sinh(\sqrt{\Delta}T)\sinh(\sqrt{\Delta}(T-t))}\right)\\
&\leq \frac{{\sinh(\sqrt{\Delta}(T-t))}\sinh(\sqrt{\Delta}T)}{1+\sinh(\sqrt{\Delta}T)\sinh(\sqrt{\Delta}(T-t))}\\
&\leq 1,\\
\end{align*}
as well as
\begin{align*}
\sqrt{\Delta}\int_0^t\frac{\sinh(\sqrt{\Delta}(T-t))}{\cosh(\sqrt{\Delta}(T-s))} ds&=\sqrt{\Delta}\int_0^t\frac{e^{-\sqrt{\Delta}(t-s)}-e^{-\sqrt{\Delta}(2T-s-t)}}{1+e^{-2\sqrt{\Delta}(T-s)}} ds\\
&\geq\sqrt{\Delta}\int_0^t\left({e^{-\sqrt{\Delta}(t-s)}-e^{-\sqrt{\Delta}(2T-s-t)}}\right)\left({1-e^{-2\sqrt{\Delta}(T-s)}} \right)ds\\
&\geq \sqrt{\Delta}\int_0^t e^{-\sqrt{\Delta}(t-s)}-e^{-\sqrt{\Delta}(2T-2t+(t-s))}-e^{-\sqrt{\Delta}(2T-2t+3(t-s))}ds\\
&\geq 1-e^{-\sqrt{\Delta}t}-2e^{-2\sqrt{\Delta}(T-t)}.
\end{align*}
In view of these two estimates, \eqref{eq:ests} yields
\begin{align}
\left|{\Delta}^{-1/2}\bar u_t\right| &\leq C_{t,T,\a,\lambda}+\Delta^{-\a/2 }R_\a\int_0^\infty e^{-u}  |u|^\a du+\sqrt{\Delta}\int_0^t e^{-\sqrt{\Delta}(t-s)}C_{s,T,\a,\lambda} ds \notag\\
&\quad+\left(e^{-\sqrt{\Delta}t}+2e^{-2\sqrt{\Delta}(T-t)}\right)\sup_{s\in[0,T]}|K^N_s| \notag\\
&\leq  4\left( e^{-\sqrt{\Delta}t}+\Delta ^{-\a/2}+e^{-2\sqrt{\Delta}(T-t)}\right)M^\a. \label{eq:estfinal}
\end{align}
(Here, the last inequality follows from the definition of $C_{\cdot,T,\a,\lambda}$.) In particular, there exists a constant $C_T>0$ only depending on $T$ such that 
\begin{align*}
\left|\int_0^T \Delta^{-1/2}\bar u_t\mu^N_tdt\right|&\leq C_T M^\a \sup_{t\in[0,T]} |\mu_t^N| (\Delta^{-\a/2}+\Delta^{-1/2})\\
&\leq \frac{C_T}{2} \left(|M^\a|^2 +\sup_{t\in[0,T]} |\mu_t^N|^2 \right)(\Delta^{-\a/2}+\Delta^{-1/2})\to 0,
\end{align*}
as $\lambda \to 0$ and in turn $\Delta \to \infty$. By the dominated convergence theorem, this pointwise convergence also holds in $L^1$, since the upper bound in this estimate is integrable under our assumptions. This shows that the Lebesgue integral in~\eqref{eq:liqcosts} is indeed of order $o(\sqrt{\lambda})$ as claimed.

The argument for the stochastic integral in~\eqref{eq:liqcosts} is similar. By the Burkholder-Davis-Gundy inequality, choosing $C_T>0$ larger if necessary, we obtain 
\begin{align*}
&\E\left[\left|\int_0^T \Delta^{-1/2}\bar u_t \sigma^N_t dW_t\right|\right]\leq C_T \E\left[\left(\int_0^T |\Delta^{-1/2}\bar u_t\sigma_t^N|^2 dt\right)^{1/2}\right]\\
&\leq 4 C_T \E\left[|M^\a|\sup_{t\in [0,T]}|\sigma_t^N|\left( \int_0^Te^{-\sqrt{\Delta}t}+\Delta ^{-\a/2}+e^{-2\sqrt{\Delta}(T-t)}dt\right)^{1/2}\right] \\
&\leq 2 C_T (\Delta^{-\a/2}+\Delta^{-1/2})^{1/2} \E\left[|M^\a|^2+\sup_{t\in [0,T]}|\sigma_t^N|^2 \right] \to 0,
\end{align*}
as $\lambda \to 0$ and in turn $\Delta \to \infty$. 
Here, we have used~\eqref{eq:estfinal} for the second inequality. Therefore, the stochastic integral in~\eqref{eq:liqcosts} also is of order $o(\sqrt{\lambda})$ in $L^1$, as $\lambda \to 0$ and the proof is complete.
\end{proof}

\bibliographystyle{abbrv}
\bibliography{bib}

\end{document}